\documentclass[sigplan, nonacm]{acmart}

\usepackage{mathpartir}
\usepackage{stmaryrd}
\usepackage{subcaption}
\usepackage{pgfplots}
\usepackage{listings}
\usepackage{tikz}
\usetikzlibrary{quantikz}
\usepackage{soul}

\lstdefinestyle{qborrow}{
  basicstyle=\ttfamily\small,
  keywordstyle=\color{blue!50!black}\ttfamily,
  commentstyle=\color{gray}\ttfamily,
  numberstyle=\color{gray}\ttfamily\tiny,
  frame=single,
  columns=fullflexible,
  keepspaces=true,
  showstringspaces=false,
  morekeywords={borrow,borrow@,X,CNOT,CCNOT}, 
  morecomment=[l]{//}, 
  numbers=left
}

\begin{document}
\title{Borrowing Dirty Qubits in Quantum Programs}
\numberwithin{equation}{section}
\numberwithin{figure}{section}
\newcommand{\bC}{{\mathbb{C}}}
\newcommand{\bF}{{\mathbb{F}}}
\newcommand{\bN}{{\mathbb{N}}}
\newcommand{\bR}{{\mathbb{R}}}
\newcommand{\Z}{{\mathbb{Z}}}
\newcommand{\cD}{\mathcal{D}}
\newcommand{\cP}{\mathcal{P}}
\newcommand{\cH}{\mathcal{H}}
\newcommand{\cE}{\mathcal{E}}
\newcommand{\cI}{\mathcal{I}}
\newcommand{\cL}{\mathcal{L}}
\newcommand{\cM}{\mathcal{M}}
\newcommand{\cX}{\mathcal{X}}
\newcommand{\cC}{\mathcal{C}}
\newcommand{\bP}{\mathbb{P}}
\newcommand{\cG}{\mathcal{G}}
\newcommand{\cF}{\mathcal{F}}
\newcommand{\cS}{\mathcal{S}}
\newcommand{\cY}{\mathcal{Y}}
\newcommand{\cZ}{\mathcal{Z}}

\newcommand{\tX}{\mathrm{X}}
\newcommand{\tH}{\mathtt{H}}
\newcommand{\tZ}{\mathtt{Z}}
\newcommand{\tCNOT}{\mathrm{CNOT}}
\newcommand{\tToffoli}{\mathrm{CCNOT}}

\newcommand{\TT}{{\mathtt{T}}}
\newcommand{\FF}{{\mathtt{F}}}

\newcommand{\sB}{\mathscr{B}}

\newcommand{\norm}[1]{\left \|{#1}\right \|}
\newcommand{\cnorm}[1]{{\left\vert\kern-0.25ex\left\vert\kern-0.25ex\left\vert #1 \right\vert\kern-0.25ex\right\vert\kern-0.25ex\right\vert}}
\newcommand{\red}[1]{{\color{purple}#1}}
\newcommand{\blue}[1]{{\color{blue}#1}}
\newcommand{\green}[1]{{\color{teal}#1}}
\newcommand{\emp}{\text{emp}}
\newcommand{\eq}[1]{(\ref{eq:#1})}
\renewcommand{\sec}[1]{Section~\ref{sec:#1}}

\renewcommand{\>}{\rangle}
\newcommand{\<}{\langle}
\newcommand{\true}{\textbf{true}}
\newcommand{\false}{\textbf{false}}
\newcommand{\Tr}{\operatorname{Tr}}
\newcommand{\rank}{\operatorname{rank}}
\newcommand{\qubits}{\mathsf{qubits}}
\newcommand{\pDensity}{\mathcal{D}^-(\mathcal{H})}
\newcommand{\sem}[1]{\left \llbracket #1 \right \rrbracket}
\newcommand{\opsem}[1]{\llparenthesis\, #1\, \rrparenthesis}
\newcommand{\lang}[1]{\left \langle #1 \right \rangle}
\newcommand{\greenb}[1]{\green{\left \{ #1 \right \}}}
\newcommand{\blueb}[1]{\blue{\left \{ #1 \right \}}}
\newcommand{\dens}[1]{{| #1 \rangle\langle #1 |}}
\newcommand{\while}[2]{\textbf{while } #1 \textbf{ do } #2 \textbf{ end}}
\newcommand{\ifel}[3]{\textbf{if } #1 \textbf{ then } #2 \textbf{ else } #3 \textbf{ end}}
\newcommand{\type}{\textbf{type}}
\newcommand{\ledis}{\mathrel{<\mkern-6mu*}}
\newcommand{\sepimp}{\mathrel{-\mkern-6mu*}}
\newcommand{\otimesimp}{\mathrel{-\mkern-3mu\otimes}}
\newcommand{\dom}{\mathop{dom}}
\newcommand{\cod}{\mathop{cod}}
\newcommand{\qreg}{{\overline{q}}}
\newcommand{\xreg}{{\overline{x}}}
\newcommand{\triple}[3]{\blueb{#1}~ #2 ~\blueb{#3}}
\newcommand{\trace}{\text{tr}}
\newcommand{\PH}{\mathcal{P}(\mathcal{H})}
\newcommand{\SPAN}{\mathop{\text{span}}}
\newcommand{\supp}[1]{{\left \lceil #1 \right \rceil}}
\newcommand{\uint}{\text{uint}}
\newcommand{\pare}[1]{\left ( #1 \right )}
\newcommand{\redp}[1]{{\color{purple}\left [ \color{black}#1 \color{purple}\right ]}}
\newcommand{\bluep}[1]{{\color{blue}\left \{ \color{black}#1 \color{blue}\right \}}}
\newcommand{\tskip}{\mathbf{skip}}
\newcommand{\tabort}{\mathbf{abort}}
\newcommand{\tif}{\mathbf{if}}
\newcommand{\tthen}{\mathbf{then}}
\newcommand{\telse}{\mathbf{else}}
\newcommand{\tfi}{\mathbf{fi}}
\newcommand{\tfor}{\mathbf{for}}
\newcommand{\twhile}{\mathbf{while}}
\newcommand{\tstate}{\texttt{State}}
\newcommand{\tdo}{\mathbf{do}}
\newcommand{\tend}{\mathbf{end}}
\newcommand{\tborrow}{\mathbf{borrow}}
\newcommand{\tfree}{\mathit{free}}
\newcommand{\tlocal}{\mathit{local}}
\newcommand{\talloc}{\mathbf{alloc}}
\newcommand{\trel}{\mathbf{release}}
\newcommand{\tcnot}{\texttt{CNOT}}
\newcommand{\tuse}{\mathbf{use}}
\newcommand{\tbegin}{\mathbf{begin}}
\newcommand{\tmaxext}{\mathbf{maxext}}
\newcommand{\fsubseteq}{{~\subseteq_{fin}~}}
\newcommand{\pluseq}{\mathrel{+}=}
\newcommand{\bwhile}{\mathbf{while}}
\newcommand{\sasaki}{\rightsquigarrow}

\newcommand{\gateX}{\mathrm{X}}
\newcommand{\gateCNOT}{\mathrm{CNOT}}
\newcommand{\gateToffoli}{\mathrm{Toffoli}}
\newcommand{\gateCCCNOT}{\mathrm{CCCNOT}}

\newcommand{\idle}{idle}
\newcommand{\tqubit}{\mathtt{qubit}}
\newcommand{\tterm}{\mathtt{term}}
\newcommand{\trexpr}{\mathtt{rexpr}}
\newcommand{\tccirc}{\mathtt{circ_{\cC}}}
\newcommand{\tclean}{\mathtt{qubit}_c}
\newcommand{\tdirty}{\mathtt{qubit}_d}
\newcommand{\trcirc}{\mathtt{rcirc}}
\newcommand{\tmidc}{\mathtt{midc}}

\begin{abstract}
\emph{Dirty qubits} are ancillary qubits that can be {borrowed} from idle parts of a computation, enabling qubit reuse and reducing the demand for fresh, clean qubits---a resource that is typically scarce in practice. For such reuse to be valid, the initial states of the dirty qubits must not affect the functionality of the quantum circuits in which they are employed. Moreover, their original states, including any entanglement they possess, must be fully restored after use---a requirement commonly known as  \emph{safe uncomputation}.
In this paper, we formally define the semantics of dirty-qubit borrowing as a feature in quantum programming languages, and introduce a notion of safe uncomputation for dirty qubits in quantum programs.  
We also present an efficient algorithm, along with experimental results, for verifying safe uncomputation of dirty qubits in certain quantum circuits.
\end{abstract}

\author{Bonan Su}
	\affiliation{
	\institution{Department of Computer Science and Technology, Tsinghua University}
	\city{Beijing}
	\country{China}
}
\email{sbn24@mails.tsinghua.edu.cn}

\author{Li Zhou}
 \affiliation{
  \institution{Key Laboratory of System Software (Chinese Academy of Sciences) and State Key Laboratory of Computer Science, Institute of Software, Chinese Academy of Sciences}   
  \city{Beijing}       
  \country{China}
}
\email{zhouli@ios.ac.cn}

\author{Yuan Feng}
\affiliation{
 \institution{Department of Computer Science and Technology, Tsinghua University}
 \city{Beijing}
 \country{China}
 }
\email{yuan_feng@tsinghua.edu.cn}

\author{Mingsheng Ying}
 \affiliation{
  \institution{Centre for Quantum Software and Information, University of Technology Sydney}          
    \city{Ultimo}
  \country{Australia}
}
\email{Mingsheng.Ying@uts.edu.au}

\maketitle 

\section{Introduction}\label{sec:introduction}
Quantum programming languages—crucial for the large-scale deployment of quantum computing—have rapidly evolved and attracted significant attention from both academia~\cite{liqui,silq,qunity,qwire,quipper} and industry~\cite{qiskit, openqasm, cirq, qsharp}.
However, quantum hardware in the current \emph{Noisy Intermediate-Scale Quantum} (NISQ) era~\cite{nisq} offers only a limited number of logical qubits and supports circuits of restricted depth and size.
These limitations are expected to persist in the near future~\cite{nisq_beyond}, posing significant challenges to realizing large-scale, practical quantum supremacy.
Therefore, it is necessary to introduce new features at the programming language level to conserve and optimize the use of these scarce quantum resources.

Introducing ancillary qubits is a crucial technique for reducing the size and depth of quantum circuits. These qubits are typically categorized as either \textit{clean qubits}, initialized in a known state such as the ground state \(|0\>\), or \textit{dirty qubits}, which may start in an arbitrary (possibly entangled) state.

Dirty qubits have been extensively studied in the design and implementation of quantum circuits, including elementary gates~\cite{elementary}, unitary synthesis~\cite{unitarysynthesis}, cryptography~\cite{crypt}, and arithmetic circuits~\cite{haner2016factoring,gidney2018factoringn2cleanqubits}. 
Figure~\ref{fig:benchmark}, adapted from Table 1 in~\cite{haner2016factoring}, summarizes the resource costs of various implementations of constant addition circuits.
While clean qubits are generally more effective for reducing circuit size or depth, dirty qubits offer greater scheduling flexibility, since the independence of their initial states allows dirty qubits to be \emph{borrowed} from any idle part of the computation---depicted in Figure~\ref{fig:borrow} with four working qubits \(q_1,q_2,q_3,q_4\).
Despite the broad investigation of dirty qubits in quantum algorithms and circuit design, relatively little attention has been given to dirty qubits from a programming language perspective, with Microsoft's Q\#~\cite{qsharp, qsharpmemory} being one of the few notable exceptions.

Unlike clean qubits, which store intermediate results in their states, correct computation with dirty qubits hinges on two key requirements:  
(1) the computation must be independent of the dirty qubits' initial states; and  
(2) the computation must \emph{safely uncompute} dirty qubits before releasing them.

For the first requirement, the \emph{toggling trick}~\cite{haner2016factoring} enables information to be propagated independently of a qubit's initial state.  
For the second requirement, analogous to the uncomputation of clean qubits~\cite{silq}, safe uncomputation of dirty qubits entails restoring each borrowed qubit---together with any entanglement it holds---to its original state after computation, which is a stricter requirement than merely restoring to the ground state for uncomputing clean qubits.


\begin{figure}[t]
    \centering
    \small
    \begin{tabular}{ccccc}
        & Cuccaro~\cite{cuccaro} & Takahashi~\cite{takahashi} & Draper~\cite{draper} & H\"aner~\cite{haner2016factoring}\\
        \hline\hline
        Size &\(\varTheta(n)\) &\(\varTheta(n)\) &\(\varTheta(n^2)\) & \(\varTheta(n\log n)\)\\
        Depth&\(\varTheta(n)\) &\(\varTheta(n)\) &\(\varTheta(n)\) &\(\varTheta(n)\) \\
        Ancillas &\(n+1\)(clean) & \(n\)(clean) & 0 & 1 (dirty)\\
        \hline
    \end{tabular}
    \caption{Costs associated with various implementations of addition \(|a\>\mapsto |a+c\>\) of a value \(a\) by a classical constant \(c\).}
    \label{fig:benchmark}
\end{figure}

\begin{figure}[t]
    \includegraphics[width=\linewidth]{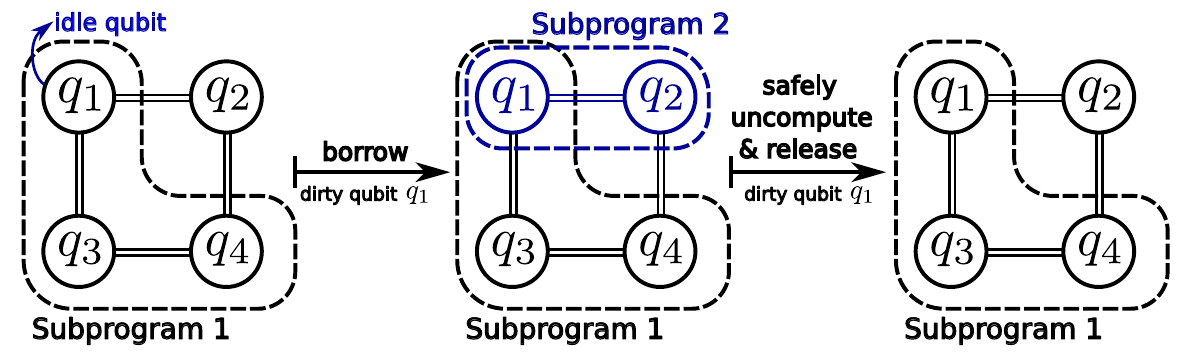}
    \caption{Borrowing idle qubits as dirty ancillas for reuse.
    }
    \label{fig:borrow}
\end{figure}


\paragraph{Difficulty in characterizing safe uncomputation of dirty qubits}
With an illustrative example, we will show that even for the most elementary circuits, formally defining safe uncomputation of dirty qubits is a non-trivial task and may be counterintuitive.
Consider the circuit for the three-controlled NOT (\(\gateCCCNOT\)) gate shown in Figure~\ref{fig:three-controlled-not}. 
It can be verified that the two circuits in Figure~\ref{fig:three-controlled-not} are \emph{equivalent}, meaning they implement the same unitary transformation regardless of the initial states of qubits \(q_1,q_2,a,q_3\) and \(q_4\). 
Therefore, a \(\gateCCCNOT\) gate can be realized using four \(\gateToffoli\) gates and one dirty qubit \(a\).

Since the circuit implements a \emph{classical} function---alternatively, it consists solely of NOT and multi-controlled NOT gates---a natural and intuitive approach to characterizing safe uncomputation is to examine all possible initial states in the computational basis, and require that the circuit restores each dirty qubit to \(|0\>\) if it was initially in \(|0\>\), and to \(|1\>\) if it was initially in \(|1\>\).

Formally, the initial state of the qubits \(q_1, q_2, a, q_3, q_4\) in the computational basis can be represented as a binary vector \((x_1, x_2, y, x_3, x_4) \in \{0,1\}^5\), and the circuit operates as a classical function \(f: \{0,1\}^5 \to \{0,1\}^5\). 
The naive definition of safe uncomputation presented in the previous paragraph would require that the final value of the dirty qubit matches its initial value for all possible inputs; that is, \(f\) safely uncomputes \(a\) if and only if the third bit of \(f(x_1, x_2, y, x_3, x_4)\) equals \(y\) for every \((x_1, x_2, y, x_3, x_4) \in \{0,1\}^5\).

Nevertheless, while this condition suffices for the uncomputation of clean qubits, it falls significantly short of capturing the safe uncomputation of dirty qubits.
A counterexample is presented in Figure~\ref{fig:unsafe-uncomputation}, where the circuit satisfies the aforementioned condition for the dirty qubit \(a\) but does not safely uncompute it. 
Particularly, the circuit fails to restore the dirty qubit \(a\) to its initial state when \(a\) is initialized in the superposition state \(|+\> = \tfrac{1}{\sqrt{2}} (|0\> + |1\>)\).

Consequently, even in elementary circuits, the safe uncomputation of dirty qubits requires a more nuanced understanding of quantum phenomena such as superposition and entanglement.
We now extend this discussion to quantum programs and formally pose the problem for borrowing and safely uncomputing dirty qubits as follows:
\begin{quote}
    \emph{What constitutes a proper formal semantics for borrowing and safely uncomputing dirty qubits in quantum programs?}
\end{quote}

\begin{figure}[t]
    \begin{minipage}{0.5\linewidth}
        \begin{quantikz}[column sep = 0.5em]
            \lstick{\(q_1\)}&\ctrl{1}&\qw\\
            \lstick{\(q_2\)}&\ctrl{2}&\qw\\
            \lstick{\(a\)}  &\qw&\qw\\
            \lstick{\(q_3\)}&\ctrl{1}&\qw\\
            \lstick{\(q_4\)}&\targ{}&\qw
        \end{quantikz}
        {\huge =}
        \begin{quantikz}[row sep = 1.2em, column sep = 0.5em]
            \qw&\ctrl{1}&\qw     &\ctrl{1}&\qw&\qw\\
            \qw&\ctrl{1}&\qw     &\ctrl{1}&\qw&\qw\\
            \qw&\targ{} &\ctrl{1}&\targ{} &\ctrl{1}&\qw\\
            \qw&\qw     &\ctrl{1}&\qw     &\ctrl{1}&\qw\\
            \qw&\qw     &\targ{} &\qw     &\targ{}&\qw
        \end{quantikz}
        \caption{Three-controlled NOT gate with dirty qubit \(a\).\\\;}
        \label{fig:three-controlled-not}
    \end{minipage}
    \hspace{0.7em}
    \begin{minipage}{0.45\linewidth}
        \centering
        \begin{quantikz}[row sep = 1.2em, column sep = 0.5em]
            \qw&\ctrl{1}&\qw     &\ctrl{1}&\qw\\
            \qw&\ctrl{1}&\qw     &\ctrl{1}&\qw\\
 \lstick{\(a\)}&\targ{} &\ctrl{1}&\targ{} &\qw\rstick{\(a\)}\\
            \qw&\qw     &\ctrl{1}&\qw     &\qw\\
            \qw&\qw     &\targ{} &\qw     &\qw
        \end{quantikz}
        \caption{Uncomputing \(a\): safe as a clean qubit but not safe as a dirty qubit.}
        \label{fig:unsafe-uncomputation}
    \end{minipage}
\end{figure}

\paragraph{Our approaches and contributions} In this paper, we propose solutions to the above problem:

\vspace{0.5\baselineskip}
\(\bullet\) \emph{Support for borrowing dirty qubits in a quantum programming language}. 
We extend the primitive quantum programming language \textsf{QWhile}~\cite{qhoare, Ying24} by introducing a new construct for borrowing and releasing dirty qubits in pairs.
Unlike existing approaches to employing dirty qubits in circuit design and optimization, our programming language dispenses with the assumption of safe uncomputation and treats dirty-qubit borrowing as a primitive construct, precisely capturing its full range of behaviors---including unsafe ones---in practice.

Inspired by the approach in~\cite{feng2023verification}, we interpret programs as sets of quantum operations and define an intuitive denotational semantics. 
Our semantics naturally accommodates the nondeterministic selection of borrowed qubits from the set of idle qubits determined by the program's syntactic structure, and also supports conditional branching and looping constructs guarded by quantum measurements.

\vspace{0.5\baselineskip}
\(\bullet\) \emph{Safe uncomputation: no action on dirty qubits.}
As illustrated in Figure~\ref{fig:three-controlled-not}, the safe uncomputation of a dirty qubit can be intuitively understood as rewriting the circuit into an equivalent diagram in which gates act only on the other qubits, while the dirty qubit is represented by a single idle wire (no-op).

Building on the formal semantic framework for our programming language with dirty-qubit borrowing, we define the safe uncomputation of a dirty qubit \(q\) in a quantum program \(S\) as the property that all executions of \(S\) are equivalent to quantum operations acting as the identity on \(q\).  
In more intuitive terms, program \(S\) safely uncomputes \(q\) if and only if executions of \(S\) are computationally indistinguishable from those quantum operations leaving \(q\) untouched. 

\vspace{0.5\baselineskip}
\(\bullet\) \emph{Efficient algorithm with experiments verifying safe uncomputation.}
For any finite-dimensional quantum program, we prove that verifying the safe uncomputation of a dirty qubit requires checking only a finite number of input-output cases. 
In the special case of quantum circuits implementing classical functions---namely, those composed solely of NOT and multi-controlled NOT gates, which encompass a broad class of circuits employing dirty qubits---we further show that verifying safe uncomputation reduces to checking the restoration of just two specific initial states of the dirty qubit. 
This reduction effectively transforms the problem into a classical satisfiability task.

By leveraging state-of-the-art \emph{Satisfiability Modulo Theories} (SMT) and \emph{Boolean Satisfiability Problem} (SAT) solvers such as CVC5~\cite{cvc5} and Bitwuzla~\cite{bitwuzla}, we successfully verify safe uncomputation in practical circuits---including constant adders and multi-controlled NOT gates---with up to hundreds and thousands of qubits, respectively.
The experimental results provide an affirmative answer to the most pressing research question in quantum verification: the scalability and efficiency of our approach.

\vspace{0.5\baselineskip}
The remainder of this paper is organized as follows.  
Section~\ref{sec:background} reviews fundamental concepts in quantum computing and programming used throughout this paper.  
As a motivating prelude, Section~\ref{sec:circuit} explores the role of dirty qubits in the design and optimization of quantum circuits.
Section~\ref{sec:proglang} introduces a quantum programming language supporting dirty-qubit borrowing, followed by Section~\ref{sec:uncomp}, which formalizes safe uncomputation of dirty qubits within this framework.  
In Section~\ref{sec:classical}, we propose an efficient algorithm for verifying safe uncomputation, accompanied by experimental results.  
Finally, we discuss potential architectural applications of dirty qubits in Section~\ref{sec:discussion} and related work in Section~\ref{sec:related}.
\section{Background on Quantum Computing}\label{sec:background}
This section provides basic background on quantum computing and introduces standard notations and principles for formally describing quantum computations using programming languages, viewed as a generalization of the circuit model.

According to the fundamental postulates of quantum mechanics, the \emph{state space} of an \(n\)-qubit quantum system is a \(2^n\)-dimensional complex Hilbert space \(\cH\). 
A \emph{quantum state} of the system is represented by a \emph{density operator} \(\rho\), which is a positive semidefinite (PSD) operator on \(\cH\) with trace equal to 1.
In particularly, we say that \(\rho\) is a \emph{pure state} if there exists a normalized vector, denoted by a \emph{ket} \(|\psi\> \in \cH\) in Dirac notation, such that \(\rho = |\psi\> \<\psi|\), where the \emph{bra} \(\< \psi|\) denotes the complex conjugate transpose of \(|\psi\>\). In this case, we may simply say that the system is in the state \(|\psi\>\).
\(|0\>\triangleq\pare{\begin{smallmatrix} 1 \\ 0 \end{smallmatrix}}\) and \(|1\>\triangleq\pare{\begin{smallmatrix} 0 \\ 1 \end{smallmatrix}}\) are the two basis states of a single qubit, with \(|0\>\) often referred to as the \emph{ground state}. \(|+\>\triangleq \frac{1}{\sqrt{2}}(|0\> + |1\>)\) is another typical one-qubit pure state, which is a linear combination of the two basis states, a phenomenon interpreted as \emph{superposition} in quantum mechanics.

\emph{Quantum operations} refer to physically feasible manipulations on the states of quantum systems, mathematically defined as a \emph{completely positive trace non-increasing linear map} on partial density operators, denoted by \(\cE: \cD(\cH) \to \cD(\cH')\). Here, ``partial'' means that we generalize the density operators to allow for a trace less than one, i.e.,
\[
\cD(\cH) = \{ A \in \cL(\cH) : A \mbox{ is PSD and}\ \Tr(A) \leq 1 \},
\]
where \(\cL(\cH)\) is the set of all linear operators on \(\cH\). 
We introduce partial density operators to describe a generalized program state that encodes termination probabilities~\cite{Ying24,qhoare}.


Given the state spaces \(\cH_1\) and \(\cH_2\) of two quantum systems, the state space of their combined system is defined as the \emph{tensor product}:
\[
\cH_1 \otimes \cH_2 \triangleq \text{span}\{|e_i\>\otimes |e_j\>: i \in I, j \in J\},
\]
where \(|e_i\>\otimes |e_j\>\) denotes the tensor product between vectors, and \(\{|e_i\>: i \in I\}\) and \(\{|e_j\>: j \in J\}\) are bases for \(\cH_1\) and \(\cH_2\), respectively.

Similarly, the product state of \(\rho_1 \in \cD(\cH_1)\) and \(\rho_2 \in \cD(\cH_2)\) is defined by the tensor product \(\rho_1 \otimes \rho_2 \in \cD(\cH_1 \otimes \cH_2)\) of the corresponding density operators.
Note that not all states in \(\cD(\cH_1 \otimes \cH_2)\) can be factorized into the tensor product of two individual states, a phenomenon known as \emph{entanglement}. 
For instance, the \emph{Bell state} \(\rho=|\Phi\rangle \langle \Phi|\), where \(|\Phi\rangle \triangleq \frac{1}{\sqrt{2}}(|0\>\otimes|0\>+|1\>\otimes|1\>)\), is a two-qubit entangled state.

The tensor product of two quantum operations \(\cE_1: \cD(\cH_1)\) \(\to \cD(\cH_1')\) and \(\cE_2: \cD(\cH_2) \to \cD(\cH_2')\) is naturally defined for tensor-factorized states as
\[
(\cE_1 \otimes \cE_2)(\rho_1 \otimes \rho_2) = \cE_1(\rho_1) \otimes \cE_2(\rho_2),
\]
and can be extended to all states by linearity.
For example, given state spaces \(\cH_1\) and \(\cH_2\) of two quantum systems, the \emph{partial trace} \(\Tr_{\cH_1}:\cD(\cH_1\otimes \cH_2)\to \cD(\cH_2)\) is a quantum operation defined as 
\[
\Tr_{\cH_1}= \Tr\otimes \cI_{\cH_2},
\]
where \(\cI_{\cH_2}\) denotes the identity operation on \(\cH_2\) and \(\Tr\) is the trace operation on \(\cH_1\).
For a state \(\rho\in \cD(\cH_1\otimes \cH_2)\) of the combined systems, the partial trace \(\Tr_{\cH_1}(\rho)\) yields the reduced state of the second system, which is obtained by tracing out the first system.

As a primitive model, this paper abstracts away types in the quantum programming language and assumes that quantum operations are applied directly to qubits. 
Specifically, a program state is a partial density operator \(\rho\in\cH_{\qubits}\triangleq\bigotimes_{q\in\qubits}\cH_q\), where \(\qubits\) is the finite set of all available logical qubits (in the quantum computer) and \(\cH_q\) is the two-dimensional state space of each qubit \(q\in\qubits\). 

A subscript \(q\) or \(\bar{q}\), representing a single qubit or a list of distinct qubits respectively, indicates that an operator---such as \(|0\>_{q}\< 0|\) or \(U_{\bar{q}}\)---acts on the specified qubits \(q\) or \(\bar{q}\).
For simplicity, the tensor product with identity operators on other qubits is sometimes omitted. 

Similarly, for a quantum operation \(\cE_{\bar{q}}\), the subscript indicates application to the qubits \(\bar{q}\), with the operation implicitly lifted to the full state space by tensoring with the identity operation \(\cI\) on the other qubits.  

We next introduce several quantum operations commonly used in quantum programming.

\paragraph{Initialization} The initialization operation initializes the qubit \(q\) to the ground state \(|0\>\), and is defined as
\[
\cE_{init,q}(\rho)\triangleq|0\>_q\<0|\rho|0\>_q\<0|+|0\>_q\<1|\rho|1\>_q\<0|.
\]
\paragraph{Unitary transformation} A \emph{unitary operator} on the Hilbert space \(\cH\) is a linear operator \(U\) such that \(U^\dagger U=UU^\dagger = I\), where \(U^\dagger\) is the adjoint operator of \(U\) and \(I\) is the identity operator on \(\cH\). 
The quantum operation for unitary transformation is defined as
\[
\cE_{U,\bar{q}}(\rho) \triangleq U_{\bar{q}} \rho U_{\bar{q}}^\dagger.
\]
Typical unitary operators mentioned in this paper include
\[
\gateX\triangleq\begin{pmatrix} 0 & 1 \\ 1 & 0 \end{pmatrix},
\gateCNOT\triangleq\pare{\begin{smallmatrix} 1 & 0 & 0 & 0 \\ 0 & 1 & 0 & 0 \\ 0 & 0 & 0 & 1 \\ 0 & 0 & 1 & 0 \end{smallmatrix}},
\gateToffoli\triangleq\pare{\begin{smallmatrix} 1 & 0 & 0 & 0 & 0 & 0 & 0 & 0 \\ 0 & 1 & 0 & 0 & 0 & 0 & 0 & 0 \\ 0 & 0 & 1 & 0 & 0 & 0 & 0 & 0 \\ 0 & 0 & 0 & 1 & 0 & 0 & 0 & 0 \\ 0 & 0 & 0 & 0 & 1 & 0 & 0 & 0 \\
0 & 0 & 0 & 0 & 0 & 1 & 0 & 0 \\ 0 & 0 & 0 & 0 & 0 & 0 & 0 & 1 \\ 0 & 0 & 0 & 0 & 0 & 0 & 1 & 0 \end{smallmatrix}}.
\]
\paragraph{Measurement} In this paper, we focus on binary measurement statements, which are specified by a pair of linear operators \(\cM = \{M_\TT, M_\FF\}\) satisfying \(M_\TT^\dagger M_\TT + M_\FF^\dagger M_\FF = I\). 
Performing the measurement \(\cM\) on a quantum state \(\rho\) has a standard physical interpretation: for each outcome \(m \in \{\TT, \FF\}\), the probability of observing \(m\) is given by \[p_m \triangleq \Tr(M_m \rho M_m^\dagger).\] When $p_m>0$, the post-measurement state becomes \[\rho_m \triangleq \frac{M_m \rho M_m^\dagger}{p_m}.\]  
By encoding the outcome probabilities into the resulting partial density operator, the measurement operation is defined as
\[
\cE_{m,\bar{q}}(\rho)\triangleq p_m\rho_m=M_m \rho M_m^\dagger,\quad
\mbox{for }m\in\{\TT,\FF\}.
\]

\section{Dirty Qubit Reuse in Quantum Circuits}\label{sec:circuit}
Before delving into the formal semantics and safety guarantees of dirty qubits in quantum programs, we first examine their proper usage in quantum circuit design and optimization.
Within this context, dirty qubits are typically utilized in two principal aspects:
\begin{enumerate}
    \item to temporarily store and propagate intermediate computation results via the toggling trick~\cite{haner2016factoring}; and
    \item once \emph{safely uncomputed}, to reuse idle qubits---including both working qubits and other dirty qubits---thereby reducing circuit width and the overall qubit requirement.
\end{enumerate}
Regarding the first aspect, several studies have investigated how to exploit dirty qubits to design circuits with the desired functionality, such as multi-controlled circuits~\cite{elementary,Gidney-blog} given in Figure~\ref{fig:three-controlled-not} and constant adder circuits~\cite{haner2016factoring}. 

In this section, we do not address how to construct such circuits; instead, we focus on the second aspect---considering, from a general perspective, what constitutes the \emph{safe uncomputation} of dirty qubits, and presenting the \emph{reuse} of qubits to reduce resource overload once dirty qubits are safely uncomputed.
As a prior, the formal definition of safe uncomputation for dirty qubits in quantum circuits is given below.
\begin{definition}[Safe Uncomputation of Dirty Qubits in Quantum Circuits~\cite{haner2016factoring}]\label{def:safe-uncomp-circuit}
    For a quantum circuit implementing the unitary transformation specified by the unitary operator \(U\), we say that a dirty qubit \(q\) is \emph{safely uncomputed} (in the circuit) if and only if
    \[
    U = V \otimes I_q \quad \text{for some unitary } V,
    \]
    where \(I_q\) denotes the identity operator on qubit \(q\).
\end{definition}

The definition states that safe uncomputation of a dirty qubit in a quantum circuit requires the circuit's unitary operator to act as the identity on that qubit, which is evidently reasonable from its equivalence to a circuit where the dirty qubit does not participate, as illustrated in Figure~\ref{fig:three-controlled-not}.

\begin{example}\label{eg:CCCNOT}
Consider the three-controlled NOT gate shown in Figure~\ref{fig:three-controlled-not}. It can be verified that 
\[
\begin{aligned}
    \gateCCCNOT_{q_1,q_2,q_3,q_4}\otimes I_{a}=\ &\gateToffoli_{a,q_3,q_4}\otimes I_{q_3,q_4}\cdot\\
    &\gateToffoli_{q_1,q_2,a}\otimes I_{q_1,q_2}\cdot\\
    &\gateToffoli_{a,q_3,a_4}\otimes I_{q_3,q_4}\cdot\\
    &\gateToffoli_{q_1,q_2,a}\otimes I_{q_1,q_2},
\end{aligned}
\]
where \(\gateCCCNOT\) denotes the three-controlled NOT gate.
Therefore, \(a\) is safely uncomputed as a dirty qubit in the circuit.
\end{example}

\begin{figure*}[t]
    \centering
    \begin{minipage}{0.38\linewidth}
        \centering
        \includegraphics[width=\linewidth]{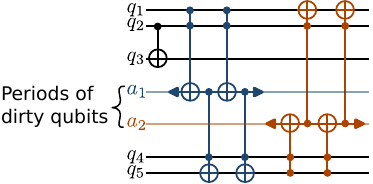}
        \subcaption{Circuit before borrowing dirty qubits.}
        \label{fig:borrowing-dirty-qubits-1}
    \end{minipage}
    \hspace{2.5em}
    \begin{minipage}{0.25\linewidth}
        \centering
        \includegraphics[width=\linewidth]{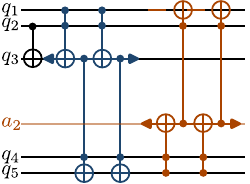}
        \subcaption{Borrowing \(q_3\) as dirty qubit \(a_1\).}
        \label{fig:borrowing-dirty-qubits-2}
    \end{minipage}
    \hspace{2.5em}
    \begin{minipage}{0.25\linewidth}
        \centering
        \includegraphics[width=\linewidth]{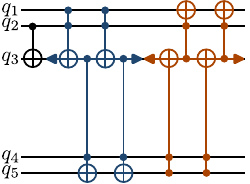}
        \subcaption{Borrowing \(q_3\) as dirty qubit \(a_2\).}
        \label{fig:borrowing-dirty-qubits-3}
    \end{minipage}
    \caption{Reducing the number of required qubits through borrowing dirty qubits.}
    \label{fig:borrowing-dirty-qubits}
\end{figure*}
For a dirty qubit \(q\), any other working qubit that remains \emph{idle} during the periods in which \(q\) is involved in the circuit and is safely uncomputed can be \emph{borrowed} to serve as \(q\). 
This eliminates \(q\) from the circuit and reduces its width. 
Moreover, unlike clean-qubit reuse, which requires verifying that the intermediate state is in the ground state, dirty-qubit reuse only requires checking the idleness of the source qubit, enabling multiple dirty qubits to be reused more freely and flexibly as the same working qubit.  
This idea is illustrated in the following example.

Consider the circuit in Figure~\ref{fig:borrowing-dirty-qubits-1}, which involves five working qubits \(q_1, q_2, q_3, q_4, q_5\) and two dirty qubits \(a_1, a_2\), implementing the routines from Figure~\ref{fig:three-controlled-not}. 
Different colors distinguish the two instances of Figure~\ref{fig:three-controlled-not} and the gates acting on the dirty qubits. 
The symbol \mbox{(\(\blacktriangleleft\!\!\!\rule[0.5ex]{0.7cm}{0.4pt}\!\!\!\blacktriangleright\))} indicates the period of each dirty qubit, during which \(a_1,a_2\) are involved in the circuit and are safely uncomputed.

During the period of \(a_1\), working qubit \(q_3\) remains idle, whereas the others are engaged in \(\gateToffoli\) gates. 
Therefore, \(q_3\) can be borrowed as the dirty qubit \(a_1\), yielding the circuit in Figure~\ref{fig:borrowing-dirty-qubits-2}. 
Similarly, since \(q_3\) is also idle during the lifetime of \(a_2\), it can be borrowed as \(a_2\) to further reduce the circuit width. 
As shown in Figure~\ref{fig:borrowing-dirty-qubits-3}, this optimization implements the same functionality of Figure~\ref{fig:borrowing-dirty-qubits-1} using only five working qubits, without any additional ancillary qubits. 

However, if \(a_1\) and \(a_2\) are clean ancillary qubits, requiring the source qubit to be in the ground state, then \(q_3\) becomes ineligible for reuse. 
This occurs because, even when \(q_3\) is initialized as a fresh clean qubit in the ground state---as typically required of source qubits in clean qubit recycling~\cite{qubit-recycling}---the leftmost \(\gateCNOT\) gate alters its state, preventing it from serving as clean qubits \(a_1\) or \(a_2\).


Reusing qubits by borrowing them as dirty qubits offers an advantage distinct from recycling clean qubits~\cite{qubit-recycling}, as it only requires the source qubit to be idle during the dirty qubits' periods, without additional computation on whether the source qubit is in the ground state before borrowing---a necessary condition for clean qubit reuse.
Consequently, dirty qubit reuse can be managed in a more general and flexible fashion by inspecting the circuit diagram, rather than by manually identifying available source qubits when designing the circuit. 
This property enables the reuse of temporarily idle working qubits as dirty qubits in complex multi-module quantum algorithms, such as factoring~\cite{haner2016factoring,gidney2018factoringn2cleanqubits} and unitary synthesis~\cite{unitarysynthesis}, thereby reducing the overall number of ancillary qubits required.
Moreover, in certain cryptographic scenarios, dirty qubits can be leveraged to minimize T-depth (or width) while still achieving some optimization of width (or T-depth)~\cite{crypt}.

\section{A Quantum Programming Language with Borrowing Dirty Qubits}\label{sec:proglang}


In this section, we elevate the concept of dirty qubit borrowing to the program level, treating it as a built-in feature of quantum programming languages.  
From a programmatic perspective, the behavior of borrowing dirty qubits in quantum programs differs from that in circuit design and optimization in the following aspects:
\begin{enumerate}
    \item Unlike in circuits, quantum programs with more complex control flow make it less straightforward to determine the active periods of dirty qubits. 
    A more programmatic approach is therefore to explicitly delineate the \emph{lifetime} of each borrowed qubit using the syntactic construct \(\tborrow\ a; \ldots; \trel\ a\), analogous to the way variable lifetimes are delineated in classical programming languages; and
    \item Instead of assuming safe uncomputation a priori, the semantics of the \(\tborrow\ a\) statement captures all possible behaviors of borrowing dirty qubits in practice---including unsafe cases. Within this semantic framework, we separately discuss the meaning of safe uncomputation of dirty qubits in quantum programs in subsequent sections.
\end{enumerate}

Formally, we extend the quantum programming language \textsf{QWhile}~\cite{qhoare, Ying24} to \textsf{QBorrow}, introducing dedicated statements for borrowing and releasing dirty qubits.  
We begin by presenting the formal syntax of \textsf{QBorrow}, followed by its denotational semantics.

\subsection{Formal Syntax}

Building on the discussion in Section~\ref{sec:background}, we use 
\[
    \qubits=\{q_1,q_2,...,q_n\}
\] 
to denote the finite set of all available logical qubits (in a quantum computer).
The syntax of \textsf{QBorrow} extends \textsf{QWhile} by introducing \(\tborrow\) statements and is formally defined by the grammar given in Figure~\ref{fig:syntax}.\footnote{The keyword \(\tborrow\) is the same as in Microsoft's Q\#~\cite{qsharp,qsharpmemory}.}

\begin{figure}[t]
    \centering
    \[
    \begin{aligned}
    \mbox{Statement }S&::=&&\tskip\\
    &\mid && [q] := |0\> \\
    & \mid && U[\bar{q}] \\
    & \mid && S_1; S_2 \\
    & \mid && \tif\ \cM[\bar{q}]\ \tthen\ S_1 \telse\ S_2\\
    & \mid && \twhile\ \cM[\bar{q}]\ \tdo\ S\ \tend\\
    & \mid && \tborrow\ a;S;\trel\ a \\
    \end{aligned}
    \]
    \caption{Syntax of \textsf{QBorrow}.}
    \label{fig:syntax}
\end{figure}

Here, each \(q\) and \(\bar{q}\) in the statements denotes a qubit or a list of distinct qubits, respectively, where a qubit may be either a concrete element of \(\qubits\) or a formal placeholder \(a\) introduced by the \(\tborrow\ a\) statements discussed in the next paragraph.
Since each unitary operator \(U\) and measurement \(\mathcal{M}\) acts on a fixed number of qubits, called its \emph{domain}, we assume that the size of \(\bar{q}\) matches the domain of \(U\) or \(\mathcal{M}\) in the constructs \(U[\bar{q}]\) and \(\mathcal{M}[\bar{q}]\), ensuring the well-formedness of the program.

In a \(\tborrow\) statement, the symbol \(a\) is not a concrete qubit from \(\qubits\), but rather a formal placeholder representing an unspecified qubit to be instantiated at runtime. 
Like a concrete qubit, it can serve as the operand of initialization, unitary, and measurement statements.
At the syntactic level, each \(\tborrow\ a\) must be paired with a matching \(\trel\ a\), thereby explicitly delimiting the lifetime of the borrowed qubit. 
For simplicity, we restrict our attention to programs in which every reference to a formal placeholder \(a\) appears within the scope of a corresponding \(\tborrow\ a\) statement, and nested \(\tborrow\) statements introduce distinct placeholders.
\begin{figure}[t]
    \[
    \begin{aligned}
        &\idle(\tskip) &&= \qubits,\\
        &\idle([q]:=|0\>) &&= \qubits\setminus\{q\},\\
        &\idle(U[\bar{q}]) &&= \qubits\setminus\bar{q},\\
        &\idle(S_1; S_2) &&= \idle(S_1)\cap \idle(S_2),\\
        &\idle(\tif\ \cM[\bar{q}]\ \tthen\ S_1 \telse\ S_2) &&= (\idle(S_1)\cap \idle(S_2))\setminus \bar{q}\\
        &\idle(\twhile\ \cM[\bar{q}]\ \tdo\ S\ \tend) &&= \idle(S)\setminus\bar{q},\\
        &\idle(\tborrow\ a;S;\trel\ a) &&= \idle(S).\\
    \end{aligned}
    \]
    \caption{Definition of idle-qubit scope.}
    \label{fig:idle}
\end{figure}

For a program statement \(S\), we denote by \(\idle(S) \subseteq \qubits\) the set of idle qubits during the execution of \(S\)---that is, the qubits not involved in any statement of \(S\).
This set is formally defined by structural induction on \(S\), as shown in Figure~\ref{fig:idle}, and specifies the qubits available for borrowing by a \(\tborrow\) statement. 
In the next subsection, we formalize the behavior of the \(\tborrow\) statement within the denotational semantics, where it is governed by the availability of idle qubits during program execution.

\subsection{Denotational Semantics}
We now formally define the denotational semantics of \textsf{QBorrow} programs. 
As one can imagine, executing a borrow statement involves nondeterministic choices of instantiations of formal placeholders. 
This makes the semantics of \textsf{QBorrow} a nontrivial generalization of \textsf{QWhile}. 

Indeed, providing a well-founded semantics for probabilistic programming languages with nondeterministic choices is a long-standing challenge~\cite{goy2020combining}; the core difficulty lies in the fundamentally different nature of the two types of choice: probabilistic choices are modeled by distributions and combined via summation over branches, whereas nondeterministic choices are modeled by sets of possible outcomes and combined via set union.  
Moreover, the absence of a distributive law between the powerset monad and the finite distribution monad leads to both conceptual complications and non-uniqueness in the resulting semantics~\cite{goy2020combining}. 
This challenge carries over \textsf{QBorrow} programs due to their probabilistic behaviors.

Nevertheless, as the primary focus of this paper is on borrowing dirty qubits rather than on semantic foundations, we follow the approach of~\cite{feng2023verification} and adopt an intuitive denotational semantics that interprets each program as a set of quantum operations.
The denotational semantics of \textsf{QBorrow} is given in Figure~\ref{fig:denotational}, where \(\sem{S}\) denotes the semantics of program \(S\) as a set of quantum operations acting on the state space \(\cD(\cH_\qubits)\).

\begin{figure}[t]
    \centering
    \[
    \begin{aligned}
        &\sem{\tskip} &&= \{ \cI \},\\
        &\sem{[q]:=|0\>} &&= \{ \cE_{init,q} \},\\
        &\sem{U[\bar{q}]} &&= \{ \cE_{U,\bar{q}} \},\\
        &\sem{S_1;S_2} &&= \{\cE_2\circ \cE_1:\cE_1\in \sem{S_1},\cE_2\in\sem{S_2}\},\\
        &\sem{\tif\ \cM[\bar{q}]\ \tthen\ S_1 \telse\ S_2}\hspace{-1000em}&& \\
        &= \{\cE_1\circ\cE_{\TT,\bar{q}}+\cE_2\circ \cE_{\FF,\bar{q}}:\cE_1\in\sem{S_1},\cE_2\in\sem{S_2}\},\hspace{-1000em}&&\\
        &\sem{\twhile\ \cM[\bar{q}]\ \tdo\ S\ \tend}\hspace{-1000em}&&\\
        &=\left \{ \sum_{n=0}^\infty \cE_{\FF,\bar{q}}\circ \bar{\cE}_n\circ \cE_{\TT,\bar{q}}\circ...\circ\bar{\cE}_{1}\circ\cE_{\TT,\bar{q}}:\bar{\cE}\in \sem{S}^\bN\right \}\hspace{-1000em}&&\\
        &\sem{\tborrow\ a;S;\trel\ a} &&=\bigcup_{q\in \idle(S)}\sem{S[q/a]}.\\
    \end{aligned}
    \]
    \caption{Denotational semantics of \textsf{QBorrow}.}
    \label{fig:denotational}
\end{figure}

In Figure~\ref{fig:denotational}, \(\cI\) denotes the identity operation, while \(\cE_{init,q}\), \(\cE_{U,\bar{q}}\), and \(\cE_{m,\bar{q}}\) represent the quantum operations defined in Section~\ref{sec:background}, corresponding respectively to initializing qubit \(q\) to the ground state \(|0\>\< 0|\), applying the unitary \(U\) to qubits \(\bar{q}\), and measuring qubits \(\bar{q}\) with outcome \(m \in \{\TT, \FF\}\).
In addition to the primitive statements \(\tskip\), initialization, and unitary transformations, we highlight the following statements:




\paragraph{Measurement-guarded branchings and loops}
The control flow in \textsf{QBorrow} is inherited from \textsf{QWhile}~\cite{qhoare,Ying24} and adheres to the principle of ``classical control, quantum data'' proposed by \citet{selinger}. 
The \(\tif\) statement is guarded by a binary measurement \(\cM[\bar{q}]\), with the measurement outcomes \(\TT\) and \(\FF\) determining which branch is executed. 

To account for all possible nondeterministic choices in \(\tif\) statements, we take the union over all schedulers that resolve the nondeterminism.
A scheduler for the \(\tif\) statement is determined by a pair of nondeterministic executions for the two branches, \(\cE_1 \in \sem{S_1}\) and \(\cE_2 \in \sem{S_2}\); thus, we take the union over all such pairs \(\cE_1, \cE_2\).

In contrast, the operations performed in different branches, \(\cE_1 \circ \cE_{\TT,\bar{q}}\) and \(\cE_2 \circ \cE_{\FF,\bar{q}}\), are combined by summation, reflecting the probabilistic nature of the program.
This contrast highlights a fundamental difference: nondeterministic choices are modeled by set union, while probabilistic choices are captured through convex combination of branches.

The \(\twhile\) statement is also guarded by a binary measurement \(\cM[\bar{q}]\), and its semantics is defined as the union over all \(\bar{\cE} \in \sem{S}^{\mathbb{N}}\), where \(\sem{S}^{\mathbb{N}}\) denotes the set of all infinite sequences of quantum operations in \(\sem{S}\), interpreted as schedulers that resolve the nondeterministic choices in each iteration of the loop body \(S\).

Mathematically, the summation series in the semantics of the \(\twhile\) statement converges because the set of quantum operations forms a \emph{complete partial order} with respect to the relation defined by
\[
\cE_1 \sqsubseteq \cE_2 \stackrel{\triangle}{\iff} \cE_2 - \cE_1 \text{ is a completely positive map},
\]
and it can be shown that the prefix sums
\[
\sum_{i=0}^n \cE_{\FF,\bar{q}} \circ \bar{\cE}_n \circ \cE_{\TT,\bar{q}} \circ \cdots \circ \bar{\cE}_1 \circ \cE_{\TT,\bar{q}}, \quad n = 0, 1, 2, \ldots
\]
form a non-decreasing sequence in \(n\).

\paragraph{Nondeterministically borrowing dirty qubits with explicit lifetime}
In the semantics of \(\tborrow\) statements, a qubit \(q \in \idle(S)\) is nondeterministically selected from the set of available idle qubits, and the program \(S\) is executed with \(q\) substituted for the formal placeholder \(a\), denoted \(S[q/a]\).  
To combine all possible nondeterministic choices of \(q\), we take the union over all \(q \in \idle(S)\) in the semantics of \(\tborrow\).
The \(\tborrow\) statement constitutes the source of nondeterminism in \textsf{QBorrow} programs. 
In particular, if \(\idle(S) = \emptyset\), meaning no idle qubits are available, then the semantics of \(\tborrow\) is the empty set, indicating that the program becomes \emph{stuck}.

\begin{figure}[t]
    \[
    \begin{aligned}
            S\equiv\quad & \gateCNOT[q_2,q_3];\\
            &\tborrow\ a_1;\\
            & S_1\begin{cases}
                & \gateToffoli[q_1,q_2,a_1];\gateToffoli[a_1,q_4,q_5];\\
                & \gateToffoli[q_1,q_2,a_1];\gateToffoli[a_1,q_4,q_5];\\
                &\tborrow\ a_2;\\
                & S_2
                    \begin{cases}
                        & \gateToffoli[q_4,q_5,q_2];\gateToffoli[a_2,q_2,q_1];\\
                        & \gateToffoli[q_4,q_5,q_2];\gateToffoli[a_2,q_2,q_1];
                    \end{cases}\\
                & \trel\ a_2;\\
            \end{cases}\\
            & \trel\ a_1.\\
    \end{aligned}
    \]
    \caption{\textsf{QBorrow} program for Figure~\ref{fig:borrowing-dirty-qubits-1} with nested \(\tborrow\) statements.}
    \label{fig:nested-borrow}
\end{figure}

As a programming-language construct, the \(\tborrow\) statement naturally permits nesting, reflecting the overlapping lifetimes of borrowed qubits.
In our semantics, borrowed qubits are drawn from \(\idle(S)\), a syntactically defined set, allowing nested \(\tborrow\) statements to borrow possibly the same qubit, which is consistent with practical implementations, as exemplified below.

Consider the program in Figure~\ref{fig:nested-borrow}, corresponding to an implementation of the circuit in Figure~\ref{fig:borrowing-dirty-qubits-1}, where \(S_1\) and \(S_2\) denote the subprograms executed during the lifetimes of \(a_1\) and \(a_2\), respectively. 
From a practical execution standpoint, working qubit \(q_3\) may be sequentially borrowed as dirty qubits \(a_1\) and \(a_2\), as shown in Figures~\ref{fig:borrowing-dirty-qubits-2} and \ref{fig:borrowing-dirty-qubits-3}.
This behavior is consistent with both the denotational semantics and the syntactic definition of idle qubits, given five working qubits \(\qubits = \{q_1, q_2, q_3, q_4, q_5\}\) in our semantics.
Since
\[
\idle(S_1) = \{q_3\} \quad \text{and} \quad \idle(S_2[q_3/a_1]) = \{q_3\},
\]
according to the denotational semantics, we have
{\small
\[
\begin{aligned}
    \sem{S}
    &=\left \{\cE \circ \cE_{\gateCNOT,q_2,q_3} : \cE \in \bigcup\nolimits_{q\in \idle(S_1)}\sem{S_1[q/a_1]}\right \}\\
    &=\{\cE \circ \cE_{\gateCNOT,q_2,q_3} : \cE \in \sem{S_1[q_3/a_1]}\} \\
    &=\left \{\cE \circ \cE_1 : \cE \in \bigcup\nolimits_{q\in \idle(S_2[q_3/a_1])}\sem{S_2[q/a_2]}\right \}\\
    &=\{\cE \circ \cE_1 : \cE \in \sem{S_2[q_3/a_1,q_3/a_2]}\} \\
    &=\{\cE_2\},
\end{aligned}
\]}

\noindent where \(\cE_1\) denotes the composition of the \(\gateCNOT\) gate and the first four \(\gateToffoli\) gates, and \(\cE_2\) denotes the final unitary transformation implemented by the circuit.

Therefore, these five qubits suffice to ensure that the program does not get stuck under the denotational semantics of \(\tborrow\) statements, faithfully reflecting the practical behavior of borrowing dirty qubits with overlapping lifetimes.

\section{Safe Uncomputation of Dirty Qubits}\label{sec:uncomp}
As a natural extension of Definition~\ref{def:safe-uncomp-circuit}, we now define the notion of safe uncomputation for dirty qubits in quantum programs.

\begin{definition}[Safe Uncomputation]\label{def:safe-uncomp}
    For a \textsf{QBorrow} program \(S\) and a qubit \(q\in \qubits\), we say that \(S\) \emph{safely uncomputes} \(q\) if and only if for all \(\cE\in \sem{S}\),
    \[
    \cE=\cI_q\otimes \cE'\mbox{ for some }\cE'\mbox{ acting on } \qubits\setminus\{q\}.
    \]
    If program \(S\) safely uncomputes qubit \(q\), then \(q\) can be substituted with a dirty qubit, and the \(\tborrow\) statement in
    \[
    \tborrow\ a;\ S[a/q];\ \trel\ a
    \]
    is said to be \emph{safe}.
\end{definition}
On the one hand, this definition is sufficiently strong to capture the essence of safe uncomputation: without inspecting intermediate states, no experimental or physical means can distinguish whether the dirty qubit participated in the execution, since the program acts as the identity on it. 
We will later demonstrate that this is equivalent to preserving the external entanglement of the dirty qubit, thereby validating the robustness of the definition in guaranteeing its safe borrowing from external systems.
On the other hand, the definition is not overly restrictive, as we will also show that it coincides with the requirement that the program restore arbitrary pure initial states---a natural and reasonably mild condition.

Moreover, we define safe uncomputation with respect to individual dirty qubits rather than entire programs, as a dirty qubit \(q\) may still be safely uncomputed even if other borrowed qubits exist that are not safely uncomputed---an idea exemplified below.
\begin{example}
    Consider the program
    \[
    S \equiv \gateX[q];\tborrow\ a;\ \gateX[q];\gateX[a];\ \trel\ a.
    \]
    Although \( \tborrow a \) is unsafe---since only an \( \gateX \) gate is applied to \( a \)---the qubit \( q \) is still safely uncomputed by \( S \). Therefore, \( q \) can still be substituted with a dirty qubit, despite the unsafe \( \tborrow a \) during its lifetime.
\end{example}

This allows the safe uncomputation of different dirty qubits to be verified individually, thereby expanding the set of safely borrowable qubits and enabling more efficient reuse of quantum resources. 

Next, we provide equivalent reinterpretations of safe uncomputation from three complementary perspectives: state restoration, entanglement preservation, and nondeterminism elimination.

\paragraph{State Restoration}
As discussed in Section~\ref{sec:introduction}, restoring only two specific states, i.e., \(|0\rangle\) and \(|1\rangle\), is insufficient to guarantee safe uncomputation, even for circuits that implement elementary classical functions.

To complement this, we further strengthen the state restoration requirement by demanding that the program restore arbitrary pure states.
This aligns with the common intuition of uncomputation: in practice, ancillary qubits are often initialized in specific pure states, 
and the circuit is expected to restore them to their initial states in a manner independent of the particular initialization.
We formally establish the equivalence between the strengthened notion of state restoration and safe uncomputation by the following theorem.
\begin{theorem}\label{thm:state-restoration}
    A program \(S\) safely uncomputes a qubit \(q\) if and only if for all states \(\rho\) and any one-qubit pure state \(|\psi\>\<\psi|\),
    \begin{equation}\label{eq:state-restoration}
    \rho|_q=|\psi\>\<\psi| \implies \forall \cE \in \sem{S}, \  \cE(\rho)|_q = |\psi\>\<\psi|.
    \end{equation}
\end{theorem}
Here, \(\rho|_q\) denotes the reduced state of \(\rho\) on qubit \(q\), obtained by tracing out all other qubits:
\[
\rho|_q \triangleq \frac{\Tr_{\qubits \setminus \{q\}}(\rho)}{\Tr(\rho)}.
\]
The implication in~\eqref{eq:state-restoration} asserts that if the initial state of \(q\) is \(|\psi\>\<\psi|\), then its final state after executing \(S\) is also \(|\psi\>\<\psi|\). 
The theorem shows that the definition of safe uncomputation is not unduly restrictive, being equivalent to restoring arbitrary pure states---a natural requirement for the safe uncomputation of dirty qubits.


\paragraph{Entanglement Preservation}
Intuitively, safely borrowing and uncomputing a dirty qubit from an external system requires, at a minimum, preserving the entanglement between the dirty qubit and the external system to which it originally belonged---referred to as \emph{external entanglement}---since otherwise, borrowing the dirty qubit could affect the external computation. 
Next, we show, via the following theorem, that safe uncomputation is in fact equivalent to preserving any external entanglement.
\begin{theorem}\label{thm:entanglement-preservation}
    A program \(S\) safely uncomputes a qubit \(q\) if and only if, for some additional hypothetical qubits \(\bar{q}' \notin \qubits\), any state \(\rho\) of \(\qubits \cup \bar{q}'\) and any entangled state \(\rho'\)
    \begin{equation}\label{eq:entanglement-preservation}
        \begin{aligned}
    &\rho|_{q,\bar{q}'}=\rho' \implies \forall \cE \in \sem{S}, \  (\cE \otimes \cI_{\bar{q}'})(\rho)|_{q,\bar{q}'}=\rho'.
        \end{aligned}
    \end{equation}
\end{theorem}
The implication in~\eqref{eq:entanglement-preservation} characterizes safe uncomputation as the preservation of external entanglement among \(q\) and \(\bar{q}'\) throughout the executions of \(S\), ensuring that \(q\) can be safely borrowed from external systems.  
Due to the presence of the \(\tborrow\ a\) statement, a program \(S\) may operate on qubits beyond those explicitly referenced within its code.  
To account for this, we introduce hypothetical qubits \(\bar{q}'\) that lies outside the scope of \(\qubits\), ensuring that \(S\) cannot act on \(\bar{q}'\).

\paragraph{Nondeterminism Elimination} 
Intuitively, the safety of a \(\tborrow\ a\) statement means that it is indistinguishable which specific qubit is borrowed as \(a\), since the program acts as the identity on the borrowed qubit, thereby eliminating the nondeterminism arising from instantiating \(a\) with different qubits.

Beyond reasoning about individual qubits, it is also important to consider the safety of an entire program.
A program \(S\) is \emph{safe} if all dirty qubits borrowed within \(S\) are safely uncomputed.
The following theorem shows that a program is safe if and only if it is equivalent to a deterministic program, eliminating all nondeterminism introduced by its \(\tborrow\) statements.
\begin{theorem}\label{thm:safe-program}
    A \textsf{QBorrow} program \(S\) is safe if and only if \(S\) is equivalent to a deterministic program.
    Formally,
    \[
    S\mbox{ is safe}\iff|\sem{S}| \leq 1\mbox{ in arbitrarily large }\qubits.
    \]
\end{theorem}
As discussed in Section~\ref{sec:proglang}, the multiplicity of \(\sem{S}\) arises from nondeterministic choices in \(\tborrow\) statements. 
Therefore, \(|\sem{S}| \leq 1\) indicates that all such choices are resolved, and the program either gets stuck (\(\sem{S}=\emptyset\)) due to an insufficient number of qubits available for borrowing, or behaves deterministically (\(|\sem{S}|=1\)), reducing to a \textsf{QWhile} program.

\section{Effective Verification of Safe Uncomputation}\label{sec:classical}
In the previous section, Theorems~\ref{thm:state-restoration} and \ref{thm:entanglement-preservation} provide characterizations of safe uncomputation from different perspectives; however, they require inspecting all pure states or entanglement, rendering the verification of safe uncomputation computationally infeasible.

In this section, we first refine these theorems to restrict the verification to a specific finite set, making the process effective. 
We then propose a reduction-to-satisfiability algorithm for circuits in which dirty qubits are commonly used, and implement and evaluate it on circuits with hundreds to thousands of qubits, thereby demonstrating the efficiency of the verification approach.

As a preparation, we introduce a finite set \(\sB\) of one-qubit density operators, forming a basis of the one-qubit state space:
\[
\sB = \{|0\>\<0|, |1\>\<1|, |+\>\<+|, |+_i\>\<+_i|\},
\]
together with one additional one-qubit pure state \(|-\>\), where the states \(|-\>\) and \(|+_i\>\) are defined as:
\[
|-\> = \tfrac{1}{\sqrt{2}}(|0\> - |1\>), \quad |+_i\> = \tfrac{1}{\sqrt{2}}(|0\> + i|1\>).
\]

The following theorem refines Theorems~\ref{thm:state-restoration} and \ref{thm:entanglement-preservation} by showing that, to verify the safe uncomputation of dirty qubits, it suffices to consider only these specific states.

\begin{theorem}[Refined Theorems~\ref{thm:state-restoration} and \ref{thm:entanglement-preservation}]\label{thm:refined}
For a \textsf{QBorrow} program \(S\) and a qubit \(q \in \qubits\), let \(n = |\qubits|\) denote the total number of qubits. Then the following conditions are equivalent:
\begin{enumerate}
    \item \(S\) safely uncomputes \(q\);
    \item \eqref{eq:state-restoration} holds for every \(|\psi\> \in \{|0\>, |1\>, |+\>, |+_i\>, |-\>\}\) and every state 
    \(\rho \in \{\rho' \otimes |\psi\>_q\<\psi| : \rho' \in \sB^{\otimes (n-1)}\}\);
    \item \eqref{eq:entanglement-preservation} holds for one additional hypothetical qubit \(\bar{q}' = q'\), the Bell state \(\rho' = |\Phi\>\<\Phi|\), and for every state \(\rho \in \{\rho' \otimes |\Phi\>_{q,q'}\<\Phi| : \rho' \in \sB^{\otimes (n-2)}\}\).
\end{enumerate}
\end{theorem}

Instead of inspecting all quantum states as required in Theorems~\ref{thm:state-restoration} and \ref{thm:entanglement-preservation}, this refined theorem reduces the verification of safe uncomputation to a finite set of states, making it computationally feasible.
Furthermore, in the following subsections, we present an algorithm that reduces the verification of safe uncomputation in certain programs to a classical Boolean satisfiability problem, and demonstrate its efficiency through experimental evaluation.

\subsection{Reduction to Boolean Satisfiability Problem}

This subsection provides a detailed explanation of how the verification of safe uncomputation in circuits implementing classical functions---in which most dirty qubits are used---can be reduced to a Boolean satisfiability problem. 

As noted in Section~\ref{sec:introduction}, even for simple circuits such as three-controlled NOT gates, ensuring the safe uncomputation of dirty qubits is non-trivial and not equivalent to uncomputation of clean qubits, requiring careful consideration of superpositions and entanglement.
We begin by presenting a necessary and sufficient condition for the safe uncomputation of dirty qubits in such circuits, as formalized in the following theorem. 

\begin{theorem}\label{thm:safe-uncomp}
    A quantum circuit \(C\) implementing a classical function (i.e., a quantum program composed solely of \(\gateX\) and multi-controlled \(\mathrm{NOT}\) gates) safely uncomputes a dirty qubit \(q \in \qubits\) if and only if the following conditions hold for all states \(\rho\):
    \[
    \begin{aligned}
        & \rho|_q = |0\>\<0| &&\implies \sem{C}(\rho)|_q = |0\>\<0|,\mbox{ and} \\
        & \rho|_q = |+\>\<+|  &&\implies \sem{C}(\rho)|_q = |+\>\<+|, \\
    \end{aligned}
    \]
    where, since \(C\) contains no \(\tborrow\) statements, \(\sem{C}\) denotes a single deterministic quantum operation without ambiguity.
\end{theorem}

This theorem refines previous Theorem~\ref{thm:state-restoration} in the scope of circuits implementing classical functions, and indicates that it suffices to check the restoration of two specific pure states: \(|0\>\) and \(|+\>\).
This result forms the foundation of our algorithm, which reduces the verification problem to a satisfiability instance.

The reduction is grounded in the fact that the semantics of a classical circuit are fully determined by its behavior on initial states in the computational basis, which corresponds to a Boolean function. 

Specifically, we associate each qubit \(q \in \qubits\) with a Boolean variable, also denoted by \(q\) without ambiguity.  
Furthermore, for each qubit \(q\), we define a Boolean formula \(b_q\) to represent and track the effect of the unitary gates in the circuit on \(q\)'s state.

Each term \(b_q\) is initially assigned the variable \(q\) and updated sequentially as follows:
\begin{itemize}
    \item Upon encountering \(\gateX[q]\), update \(b_q := \neg b_q\);
    \item Upon encountering an \(m\)-controlled NOT gate 
    \[
    \mathrm{C}^m\mathrm{NOT}[q_1,...,q_m,q_{m+1}],
    \] 
    update 
    \[
        b_{q_{m+1}}:=b_{q_{m+1}}\oplus (b_{q_1}b_{q_2}\cdots b_{q_m}),
    \]
\end{itemize}
where \(\oplus\) denotes XOR, and the concatenation \(b_{q_1}b_{q_2}\cdots b_{q_m}\) suppresses the explicit AND operator for brevity.
\begin{figure}[t]
    \[
    \begin{array}{cccccc}
            & b_{q_1} & b_{q_2} & a & b_{q_3} & b_{q_4}\\[0.3em]
        \hline\noalign{\vskip 2pt}
        \mbox{Initial}  & q_1 & q_2 & a & q_3 & q_4\\
        \mbox{1st gate} & q_1 & q_2 & a\oplus q_1q_2 & q_3 & q_4\\
        \mbox{2nd gate} & q_1 & q_2 & a\oplus q_1q_2 & q_3 & q_4\oplus q_3(a\oplus q_1q_2)\\
        \mbox{3rd gate} & q_1 & q_2 & a & q_3 & q_4\oplus q_3(a\oplus q_1q_2)\\
        \mbox{4th gate} & q_1 & q_2 & a & q_3 & q_4\oplus q_3(a\oplus q_1q_2)\oplus q_3a
    \end{array}
    \]
    \caption{Construction of Boolean formulas for a three-controlled NOT gate.}
    \label{fig:boolean-formula}
\end{figure}
\begin{example}
    Consider the right circuit in Figure~\ref{fig:three-controlled-not}, where \(a\) is treated as a concrete qubit.
    The updates to each Boolean formula according to the rules described above are summarized in Figure~\ref{fig:boolean-formula},
    where we simplify \(b_a = a \oplus q_1q_2 \oplus q_1q_2\) to \(b_a = a\) after the third gate, using the identity \(x \oplus x = 0\).
\end{example}

After constructing the Boolean formula \(b_q\) for each qubit, the first condition in Theorem~\ref{thm:safe-uncomp} reduces to checking the unsatisfiability of the Boolean formula  
\begin{equation}\label{eq:cond-1}
    \neg(b_q\to q),
\end{equation}  
where \(\to\) denotes logical implication.  
By the truth table of implication, \eqref{eq:cond-1} is unsatisfiable if and only if \(q = 0\) implies \(b_q = 0\).

The second condition in Theorem~\ref{thm:safe-uncomp} requires a more refined analysis and is equivalent to verifying the unsatisfiability of the Boolean formula  
\begin{equation}\label{eq:cond-2}
\bigvee_{q' \in \qubits, q' \neq q} b_{q'}[0/q] \oplus b_{q'}[1/q].
\end{equation}  

Intuitively, restoring the state \(|+\>\) requires that the final states of all other qubits be identical whether \(q\) is in \(|0\>\) or \(|1\>\); equivalently, their final states must be independent of \(q\), precisely capturing the essence of safe uncomputation.

\begin{theorem}\label{thm:boolean}
    A quantum circuit \(C\) implementing a classical function safely uncomputes a dirty qubit \(q\) if and only if both Boolean formulas \eqref{eq:cond-1} and \eqref{eq:cond-2} are unsatisfiable. 
\end{theorem}

\begin{figure}[t]
    \centering
    \begin{lstlisting}[style=qborrow]
 // adder.qbr
 let n = 50;        // number of qubits
 borrow@ q[n];      // skip verification
 borrow a[n - 1];   // dirty qubits
 CNOT[a[n - 1], q[n]];
 for i = (n - 1) to 2 {
     CNOT[q[i], a[i]];
     X[q[i]];
     CCNOT[a[i - 1], q[i], a[i]];
 }
 CNOT[q[1], a[1]];
 for i = 2 to (n - 1) {
     CCNOT[a[i - 1], q[i], a[i]];
 }
 CNOT[a[n - 1], q[n]];
 X[q[n]];

 // reverse the circuit to uncompute
 for i = (n - 1) to 2 {
     CCNOT[a[i - 1], q[i], a[i]];
 }
 CNOT[q[1], a[1]];
 for i = 2 to (n - 1) {
     CCNOT[a[i - 1], q[i], a[i]];
     X[q[i]];
     CNOT[q[i], a[i]];
 }
    \end{lstlisting}
    \caption{Implementation of adder circuits in \textsf{QBorrow}.}
    \label{fig:adder-circuit}
\end{figure}

\begin{figure}[t]

    \centering
    \begin{tikzpicture}
        \begin{axis}[
            xlabel=Number of qubits,
            ylabel=Verification Duration (s),
            ymode=linear,
            ytick={0,150, 300, 450, 600, 750, 900, 1050, 1200},
            xtick distance=25,
            width=8cm,
            height=5.5cm,
            legend pos=north west
        ]
        \addplot[blue, mark=x] coordinates {
            (50, 4)
            (75, 24)
            (100, 71)
            (125, 171)
            (150, 365)
            (175, 751)
            (200, 1069)
        };
        \addlegendentry{CVC5}

        \addplot[red, mark=+] coordinates {
            (50, 3)
            (75, 12)
            (100, 29)
            (125, 98)
            (150, 158)
            (175, 248)
            (200, 313)
        };
        \addlegendentry{Bitwuzla}

        \end{axis}
    \end{tikzpicture}

    \caption{Verification overhead of adder circuits.}
    \label{fig:adder-result}

~\\
    \begin{tikzpicture}
        \begin{axis}[
            xlabel=Number of qubits,
            ylabel=Verification Duration (s),
            ymode=linear,
            ytick={0,30, 60, 90, 120, 150, 180, 210, 240},
            xtick distance=500,
            width=8cm,
            height=5.5cm,
            legend pos=north west
        ]
        \addplot[blue, mark=x] coordinates {
            (499, 0)
            (999, 1)
            (1499, 4)
            (1999, 7)
            (2499, 11)
            (2999, 17)
            (3499, 27)
        };
        \addlegendentry{CVC5}

        \addplot[red, mark=+] coordinates {
            (499, 3)
            (999, 16)
            (1499, 35)
            (1999, 61)
            (2499, 115)
            (2999, 163)
            (3499, 239)
        };
        \addlegendentry{Bitwuzla}

        \end{axis}
    \end{tikzpicture}

    \caption{Verification overhead of MCX circuits.}
    \label{fig:mcx-result}
\end{figure}

\subsection{Implementation and Evaluation}


To address the most central concerns in the field of quantum circuit verification---
\begin{quote}
\emph{Is our approach \emph{efficient} and \emph{scalable} enough to verify large-scale quantum circuits?}
\end{quote}
We implement the reduction algorithm described in the previous subsection and evaluate it on two benchmark programs. 
We first present the implementation methodology, followed by a discussion of the experimental results.

\paragraph{Implementation Methodology}
A restricted version of the \textsf{QBorrow} language is implemented in C++, supporting \(\tborrow\) statements, basic \(\tfor\) loops, and unitary gate applications. 
The frontend parser, responsible for syntax analysis and error handling, is generated using ANTLR4---a widely used parser generator for processing structured text.

A sample benchmark program is shown in Figure~\ref{fig:adder-circuit}, which implements a constant adder circuit adapted from Fig.~2 in~\cite{haner2016factoring}. 
For an ideal experiment, the statement \(\tborrow @\ q[n]\) indicates that no assumptions are made about the initial states of \(n\) qubits \(q[1], \ldots, q[n]\); equivalently, the program borrows \(q[1],...,q[n]\) without requiring safely uncomputing them. 
It can be verified that the circuit acts as the identity on qubits \(a[1], \ldots, a[n]\) and \(q[1], \ldots, q[n{-}1]\), while writing the most significant bit of the sum \((s_1 \ldots s_n)_2 + (11\ldots 1)_2\) into \(q[n]\), where the input bits \(s_1 \ldots s_n\) are encoded in the states of \(q[1], \ldots, q[n{-}1]\).

The backend verifier translates classical circuits into Boolean formulas whose satisfiability must be checked, as described in the previous subsection. 
The resulting formulas are submitted to state-of-the-art SMT and SAT solvers, CVC5~\cite{cvc5} and Bitwuzla~\cite{bitwuzla}, for satisfiability checking.

\paragraph{Experimental Results}
In addition to the adder circuit, we also evaluate our implementation on multi-controlled NOT (MCX) circuits adapted from~\cite{Gidney-blog}, where a \((2n-1)\)-controlled NOT gate is implemented using one dirty ancilla qubit and \(16(n-2)\) \(\gateToffoli\) gates for natural numbers \(n\geq 3\).
The full implementation of the MCX circuit in \textsf{QBorrow} is available in the supplementary material.

All benchmarks were executed on a MacBook Air equipped with an 8-core Apple M3 processor and 24 GB of RAM. 
The code was compiled using the GNU \texttt{g++} compiler with optimization level \texttt{-O3} enabled.

Figures~\ref{fig:adder-result} and~\ref{fig:mcx-result} summarize the experimental results for verifying the safe uncomputation of \(n\) dirty qubits in the adder circuit and a single dirty qubit in the MCX circuits, respectively. 
The verification time includes only the duration taken by the SMT solvers to check the satisfiability of the resulting formulas, which constitutes the main performance bottleneck.
It excludes parsing, preprocessing, debugging output, and the reduction process, as the construction of Boolean formulas involves only a linear scan of the circuit and completes in under one second.

The experimental results demonstrate that our reduction method is both effective and scalable, enabling efficient verification of adder circuits with hundreds of qubits, as the verification overhead grows polynomially with the number of qubits.
For structurally simpler MCX circuits, the method can handle instances with thousands of qubits.  
Notably, due to differences in incremental solving strategies and formula simplification algorithms, CVC5 and Bitwuzla each exhibit efficiency advantages on one of the two benchmark programs.

\section{Discussions}\label{sec:discussion}
We now discuss potential architectural applications of dirty qubits and the corresponding issues related to verifying their safety.

\paragraph{Single-program optimization\cite{gidney2018factoringn2cleanqubits,haner2016factoring,crypt}}
Existing quantum programming languages support implicit utilization of dirty qubits, where programmers manually specify variables to participate in the computations, even if they act as dirty qubits.
However, the programmer may mistakenly identify a qubit as idle when it is not, or fail to account for qubits that only become idle after compilation and gate parallelization. 
Therefore, dirty qubit scheduling is better handled by the compiler, relieving the programmer of this burden and enabling more accurate and efficient qubit reuse.

Microsoft's Q\#~\cite{qsharp,qsharpmemory} language provides a \(\tborrow\) statement that allows programmers to borrow qubits from other parts of the program. 
However, it has not yet provided the definition and verification of safe uncomputation, leaving all potential risks to the programmer.  
While this approach is acceptable for qubit reuse within the optimization of a single program, it becomes extremely risky in the multi-program scenarios discussed next.

\paragraph{Multi-program scheduling~\cite{multiprogram}}\footnote{A program that manages and schedules quantum resources across various computation tasks on the same quantum processor.}
Multi-programming is a technique used in quantum cloud services like QuCloud \cite{QuCloud,QuCloud+}  to enhance the throughput and utilization of NISQ machines by running multiple workloads concurrently. 
By borrowing idle qubits between different programs or rescheduling program priorities to manage qubit borrowing, multiple programs can still take advantage of the reductions in circuit size and depth provided by dirty qubits under concurrent execution.

However, this approach requires a higher level of safety assurance, since incorrectly returning a borrowed dirty qubit---whether through state alteration or entanglement disruption---can cause errors or even crashes in other programs.  
Beyond the safe uncomputation verification addressed in this paper, termination analysis is also necessary, as programs that borrow dirty qubits but fail to terminate are regarded as exhibiting improper behavior.
The analysis of program termination is independent of and complementary to the safe uncomputation studied in this paper. Furthermore, several existing studies have already explored this area~\cite{termination-problem,algebraic-termination}.
\section{Related Work}\label{sec:related}
\paragraph{Utilizations of dirty qubits}
\citet{elementary} were the first to investigate the construction of multi-controlled X (MCX) gates using dirty qubits.  
Subsequently, \citet{haner2016factoring} and \citet{gidney2018factoringn2cleanqubits} studied the use of dirty qubits to optimize resource costs in quantum circuits for factorization.  
More recently, \citet{unitarysynthesis} proposed a general method to trade \(\mathrm{T}\) gates for dirty qubits in universal state preparation and unitary synthesis, and \citet{crypt} employed dirty qubits to minimize the \(T\)-depth or width of quantum circuit gadgets for cryptographic applications.

In the realm of programming languages, ReQwire \cite{reqwire} noted that uncomputing dirty ancilla qubits requires substantial additional mechanisms but offers significant gains in expressiveness, and languages such as Q\# \cite{qsharp} already facilitate the borrowing of dirty qubits.
However, none of these works addressed the formal characterization and verification of safe uncomputation of dirty qubits.

\paragraph{Uncomputation of clean qubits} 
Automatic uncomputation improves both resource management and the usability of quantum programming. 
Silq~\cite{silq} was the first quantum language to provide safe, automatic uncomputation via type checking, simplifying the handling of temporary values and reducing code complexity and programmer effort.  
\citet{unqomp} introduced Unqomp, the first procedure to automatically synthesize uncomputation in quantum circuits. Building on Unqomp, \citet{qrispuncomp} integrated uncomputation into the Qrisp high-level quantum programming framework. 
\citet{reqomp} further explored space-constrained uncomputation, optimizing quantum circuits by trading qubits for gates. 
Most recently, \citet{modularuncomp} introduced the first modular approach to automatic uncomputation for expressive quantum programs, using modular algorithms over an intermediate representation to enhance efficiency and scalability.
However, none of these works addressed the safe uncomputation or efficient synthesis of dirty qubits, making it an attractive direction to extend or adapt these techniques to support dirty qubits.

\newpage

\bibliographystyle{ACM-Reference-Format}
\bibliography{ref.bib}

\newpage
\section{Deferred Proofs}

\subsection{Proof of Theorem~\ref{thm:state-restoration}}
\begin{theorem}
    A program \(S\) safely uncomputes a qubit \(q\) if and only if for all states \(\rho\) and any one-qubit pure state \(|\psi\>\<\psi|\),
    \begin{equation}
    \rho|_q=|\psi\>\<\psi| \implies \forall \cE \in \sem{S}, \  \cE(\rho)|_q = |\psi\>\<\psi|.
    \end{equation}
\end{theorem}
\begin{proof}
    (\(\Rightarrow\)) Suppose that the program \(S\) safely uncomputes the qubit \(q\). It follows from the definition that for any \(\cE\in \sem{S},\ \cE=\cI_q\otimes \cE'\) for some \(\cE'\).
    Therefore, when \(\rho|_q=|\psi\>\<\psi|\) for arbitrary one-qubit pure state \(|\psi\>\),  \(\rho=\rho'\otimes|\psi\>_q\<\psi|\) for some \(\rho'\) and 
    \[
    \cE(\rho)=(\cE'\otimes \cI_q)(\rho'\otimes |\psi\>_q\<\psi|)=\cE'(\rho')\otimes |\psi\>_q\<\psi|,
    \]
    which implies that \(\cE(\rho)|_q=|\psi\>\<\psi|\).

    (\(\Leftarrow\)) For each quantum operation \(\cE\in\sem{S}\), it suffices to show that if for all pure states \(|\psi\>_q\<\psi|\otimes |\varphi\>\<\varphi|\), we have 
    \[
    \cE(|\psi\>_q\<\psi|\otimes |\varphi\>\<\varphi|) = |\psi\>_q\<\psi|\otimes \rho\mbox{ for some }\rho,
    \]
    then \(\cE=\cI_q\otimes \cE'\) for some \(\cE'\).

    Let \(\cH=\bigotimes_{q'\in \qubits\setminus \{q\}}\cH_{q'}\), consider the Stinespring dilation of the quantum operation \(\cE\), there exists another Hilbert space $\cH'$ and a linear operator $A\in\cL(\cH_q\otimes\cH,\cH_q\otimes\cH\otimes \cH')$ such that 
    \[
    \cE(|\psi\>_q\<\psi|\otimes|\varphi\>\<\varphi|)=\Tr_{\cH'}\pare{A(|\psi\>_q\<\psi|\otimes |\varphi\>\<\varphi|)A^\dagger}.
    \]
    Because \(\cE(|\psi\>_q\<\psi|\otimes|\varphi\>\<\varphi|)=|\psi\>_q\<\psi|\otimes \rho\), it can be verified that 
    \[
    A|\psi\>|\varphi\>=|\psi\>|\phi\>
    \]  
    for some \(|\phi\>\in \cH\otimes \cH'\) such that \(\Tr_{\cH'}(|\varphi\>\<\varphi|)=\rho\).
    Since \(|\psi\>\in\cH_q\) is arbitrarily chosen, there exists another linear operator $B\in \cL(\cH,\cH\otimes\cH')$ such that \(A=I_q\otimes B\). We then construct \(\cE'\) as 
    \[
    \cE'(\rho)=\Tr_{\cH'}(B\rho B^\dagger),
    \]
    from the property of Stinespring dilation, we can conclude that \(\cE=\cI_q\otimes \cE'\).
\end{proof}

\subsection{Proof of Theorem~\ref{thm:entanglement-preservation}}
\begin{theorem}
    A program \(S\) safely uncomputes a qubit \(q\) if and only if, for some additional hypothetical qubits \(\bar{q}' \notin \qubits\), any state \(\rho\) of \(\qubits \cup \bar{q}'\) and any entangled state \(\rho'\)
    \begin{equation}
        \begin{aligned}
    &\rho|_{q,\bar{q}'}=\rho' \implies \forall \cE \in \sem{S}, \  (\cE \otimes \cI_{\bar{q}'})(\rho)|_{q,\bar{q}'}=\rho'.
        \end{aligned}
    \end{equation}
\end{theorem}
\begin{proof}
    (\(\Rightarrow\)) Suppose that the program \(S\) safely uncomputes the qubit \(q\). 
    It follows from the definition that for any \(\cE\in \sem{S},\ \cE=\cI_q\otimes \cE'\) for some \(\cE'\). 
    Given a state \(\rho\) with decomposition \(\rho=\sum_i\alpha_i(\rho^{(i)}\otimes \sigma_{q,\bar{q}'}^{(i)})\) (note that \(\alpha_i\) are complex coefficients), we have
    \[
    \begin{aligned}
        (\cE\otimes\cI_{\bar{q}'})(\rho)|_{q,\bar{q}'}&=\frac{\Tr_{\qubits\setminus \{q\}\cup \bar{q}'}\pare{(\cE'\otimes \cI_{q,\bar{q}'})(\rho)}}{\Tr((\cE'\otimes \cI_{q,\bar{q}'})(\rho))}\\
        &=\frac{\pare{(\Tr\circ\cE')\otimes \cI_{q,\bar{q}'}}(\rho)}{\Tr((\cE'\otimes \cI_{q,\bar{q}'})(\rho))}\\
        &=\frac{\sum_i\alpha_i\pare{(\Tr\circ\cE')\otimes \cI_{q,\bar{q}'}}(\rho^{(i)}\otimes \sigma^{(i)}_{q,\bar{q}'})}{\sum_i\alpha_i\Tr((\cE'\otimes \cI_{q,\bar{q}'})(\rho^{(i)}\otimes \sigma^{(i)}_{q,\bar{q}'}))}\\
        &=\frac{\sum_i\alpha_i\Tr(\cE'(\rho^{(i)}))\otimes \sigma^{(i)}_{q,\bar{q}'}}{\sum_i\alpha_i\Tr(\cE'(\rho^{(i)}))\otimes \Tr(\sigma^{(i)}_{q,\bar{q}'})}\\
        &=\frac{\sum_i\alpha_i\sigma^{(i)}}{\Tr(\sum_i\alpha^{(i)})}\\
        &=\rho|_{q,\bar{q}'}.
    \end{aligned}
    \]
\end{proof}

\subsection{Proof of Theorem~\ref{thm:safe-program}}
\begin{theorem}
    A \textsf{QBorrow} program \(S\) is safe if and only if \(S\) is equivalent to a deterministic program.
    Formally,
    \[
    S\mbox{ is safe}\iff|\sem{S}| \leq 1\mbox{ in arbitrarily large }\qubits.
    \]
\end{theorem}
\begin{proof}
    (\(\Rightarrow\)) Assume that \(S\) is safe and \(\qubits\) is sufficiently large, ensuring that every \(\tborrow\) statement in \(S\) has at least one option. Consequently, \(S\) will not get stuck, and \(\sem{S}>0\).
    We prove by induction on the structure of \(S\) that \(|\sem{S}|=1\).
    In addition, we strengthen the induction hypothesis by requiring not only that \(|\sem{S}|=1\), but also that \(\sem{S}\) acts as identity on qubits other than \(qubits(S)\).

    The base case is when \(S\equiv \tskip,[q]:=|0\>\) or \(U[\bar{q}]\), by the denotational semantics given in Figure~\ref{fig:denotational}, it trivially holds that \(|\sem{S}|=1\) and for all \(q' \in \qubits\setminus qubits(S)\), \(\sem{S}\) acts as identity on \(q'\).

    If \(S\equiv S_1;S_2\) and \(S\) is safe, then \(S_1,S_2\) are safe programs. 
    By induction hypothesis, \(|\sem{S_1}|=|\sem{S_2}|=1\).
    Therefore, 
    \[
    |\sem{S}|=|\cE_1\circ\cE_2:\cE_1\in \sem{S_2},\cE_2\in\sem{S_1}|=|\sem{S_1}|\cdot|\sem{S_2}|=1.
    \]
    Furthermore, for qubit
    \[
    q\in\idle(S_1;S_2)=\idle(S_1)\cap \idle(S_2),
    \]
    by induction hypothesis we know that \(\sem{S_1}\) and \(\sem{S_2}\) act as identity on \(q\), thus \(\sem{S}\) also acts as identity on \(q\).
    A similar proof holds for \(S\equiv\tif\ \cM[\bar{q}]\ \tthen\ S_1\ \telse\ S_2\) or \(S\equiv \twhile\ \cM[\bar{q}]\ \tdo\ S\ \tend\).

    If \(S \equiv \tborrow\ a\ S'\ \trel\ a\), then, since \(S\) is safe, for every qubit \(q \in \idle(S')\), the program \(S'[q/a]\) is also safe. 
    By the induction hypothesis, this implies \(|\sem{S'[q/a]}| = 1\).  
    Moreover, as \(S\) is safe, each borrowed qubit \(q\) is safely uncomputed in \(S'[q/a]\), and thus \(\sem{S'[q/a]}\) acts as the identity on both \(q\) and all qubits in \(\idle(S')\). 
    It follows that the semantics \(\sem{S'[q/a]}\) are identical for all \(q \in \idle(S')\).
    Therefore, we can conclude that \(|\sem{S}|=|\bigcup_q\sem{S'[q/a]}|=1\).

    \((\Leftarrow)\) We proceed by proving the contrapositive: if \(|\sem{S}| > 1\) and \(\qubits\) is sufficiently large so that every \(\tborrow\) statement in \(S\) has at least two choices, then some qubit borrowed in \(S\) is unsafe.

    According to the denotational semantics in Figure~\ref{fig:denotational}, any non-determinism in \(\sem{S}\) can only originate from the \(\tborrow\) statements within \(S\). In particular, if \(|\sem{S}| > 1\), then there exists a substatement of the form \(\tborrow\ a;\ S'\ \trel\ a\) in \(S\) and two distinct qubits \(q_1, q_2 \in \idle(S')\) such that
    \[
    \sem{S'[q_1/a]} \neq \sem{S'[q_2/a]}.
    \]
    This implies that the borrowed qubit is not safely uncomputed, as otherwise both instantiations would act as the identity on \(q_1\) and \(q_2\), resulting in identical semantics.
\end{proof}

\subsection{Proof of Theorem~\ref{thm:refined}}
\begin{theorem}[Refined Theorems~\ref{thm:state-restoration} and \ref{thm:entanglement-preservation}]
For a \textsf{QBorrow} program \(S\) and a qubit \(q \in \qubits\), let \(n = |\qubits|\) denote the total number of qubits. Then the following conditions are equivalent:
\begin{enumerate}
    \item \(S\) safely uncomputes \(q\);
    \item \eqref{eq:state-restoration} holds for every \(|\psi\> \in \{|0\>, |1\>, |+\>, |+_i\>, |-\>\}\) and every state 
    \(\rho \in \{\rho' \otimes |\psi\>_q\<\psi| : \rho' \in \sB^{\otimes (n-1)}\}\);
    \item \eqref{eq:entanglement-preservation} holds for one additional hypothetical qubit \(\bar{q}' = q'\), the Bell state \(\rho' = |\Phi\>\<\Phi|\), and for every state \(\rho \in \{\rho' \otimes |\Phi\>_{q,q'}\<\Phi| : \rho' \in \sB^{\otimes (n-2)}\}\).
\end{enumerate}
\end{theorem}
\begin{proof}
    (\(1\implies 2\)) Similar to the proof of (\(\Rightarrow\)) direction in Theorem~\ref{thm:state-restoration}.

    (\(2\implies 1\)) Let \(\cH=\bigotimes_{q'\in\qubits\setminus \{q\}}\cH_{q'}\). 
    Given that \(\sB^{\otimes (n-1)}\) forms a basis of \(\cL\pare{\cH}\), it suffices to show that for a quantum operation \(\cE\), if the following condition holds for all states \(\rho \in \cD(\cH)\) and \(|\psi\> \in \{|0\>, |1\>, |+\>, |+_i\>, |-\>\}\):  
    \[
    \cE(\rho \otimes |\psi\>_q\<\psi|) = \rho' \otimes |\psi\>_q\<\psi| \quad \text{for some } \rho',
    \]  
    then it follows that \(\cE = \cI_q \otimes \cE'\) for some \(\cE'\).

    We know that \(\{|0\>\<0|,|1\>\<1|,|+\>\<+|,|+_i\>\<+_i|\}\) forms a basis of \(\cL(\cH_q)\) and
    \[
    \begin{aligned}
        &\cE(\rho\otimes |0\>_q\<0|)=\rho_0\otimes |0\>_q\<0|,\\
        &\cE(\rho\otimes |1\>_q\<1|)=\rho_1\otimes |1\>_q\<1|,\\
        &\cE(\rho\otimes |+\>_q\<+|)=\rho_+\otimes |+\>_q\<+|,\\
        &\cE(\rho\otimes |+_i\>_q\<+_i|)=\rho_{+_i}\otimes |+_i\>_q\<+_i|,\\
        & \cE(\rho\otimes |-\>_q\<-|)=\rho_-\otimes |-\>_q\<-|.
    \end{aligned}
    \]
    Observe that \(|-\>\<-|\) is an operator with non-zero overlap with each of the other four basis operators \(|0\>\<0|, |1\>\<1|, |+\>\<+|\), \(|+_i\>\<+_i|\). Hence, it follows that \(\rho_0 = \rho_1 = \rho_+ = \rho_{+_i} = \rho_-\).  
    In other words, \(\rho_0, \rho_1, \rho_+, \rho_{+_i}\) are independent of the choice of \(|\psi\>\) and depend solely on \(\rho\).

    Since \(\{|0\>\<0|,|1\>\<1|,|+\>\<+|,|+_i\>\<+_i|\}\) forms a basis of \(\cL(\cH_q)\), we know that for all one-qubit pure states \(|\psi\>\in \cH_q\),
    \[
    \cE(\rho\otimes |\psi\>_q\<\psi|)=\rho'\otimes |\psi\>_q\<\psi|\mbox{ for some }\rho',
    \]
    since one may decompose \(|\psi\>_q\<\psi|\) into the linear combination of the four basis operators.

    By applying the proof of the \(\Leftarrow\) direction in Theorem~\ref{thm:state-restoration}, we conclude that \(\cE = \cI_q \otimes \cE'\)

    (\(1\implies 3\)) Similar to the proof of (\(\Rightarrow\)) direction in Theorem~\ref{thm:entanglement-preservation}.

    (\(3\implies 1\)) For each \(\cE\in\sem{S}\), consider the Kraus representation \(\cE(\rho)=\sum_kE_k\rho E_k^\dagger\) and an orthonormal basis \(\{|i\>\}\) of \(\bigotimes_{q''\in\qubits\setminus\{q\}}\cH_{q''}\), we have 
    \[
    (\cE\otimes\cI_{q'})(|i\>\<i|\otimes |\Phi\>_{q,q'}\<\Phi|)=\rho'\otimes |\Phi\>_{q,q'}\<\Phi|\mbox{ for some }\rho',
    \]
    where the left-hand side equals to 
    \[
    \begin{aligned}
        &(\cE\otimes\cI_{q'})(|i\>\<i|\otimes |\Phi\>_{q,q'}\<\Phi|) \\
        =& \sum_k (E_k\otimes I_{q'})(|i\>\<i|\otimes |\Phi\>_{q,q'}\<\Phi|)(E_k^\dagger\otimes I_{q'})\\
        =&\frac{1}{2}\sum_k (E_k\otimes I_{q'})(|i\>\<i|\otimes (|00\>+|11\>)_{q,q'}(\<00|+\<11|))(E_k^\dagger\otimes I_{q'}).\\
    \end{aligned}
    \]
    Therefore, for each \(i,k\), there exists a vector \(|\beta_{i,k}\>\) such that 
    \[
    (E_k\otimes I_{q'})(|i\>\otimes (|00\>+|11\>)_{q,q'})=|\beta_{i,k}\>\otimes (|00\>+|11\>)_{q,q'}.
    \]
    For $j=0,1$, by multiplying $I\otimes \<j|_{q'}$ on both sides of the equation we can derive that
    $$
    E_k(|i\>\otimes |j\>_q)=|\beta_{k,i}\>\otimes|j\>_q.
    $$
    Therefore, $E_k=\sum_{i,j}|\beta_{k,i}\>\<i|\otimes|j\>_q\<j|=E_k'\otimes I_q$ for some $E_k'$.
    And we can conclude that $\cE=\cE'\otimes\cI_q$ for some $\cE'$.

    In fact, the proof remains valid if the state \(|\Phi\rangle\) is replaced with an arbitrary entangled state of the form \(\alpha|00\rangle + \beta|11\rangle\), where \(|\alpha|, |\beta| \neq 0\).
    
\end{proof}

\subsection{Proof of Theorem~\ref{thm:safe-uncomp}}
\begin{theorem}
    A quantum circuit \(C\) implementing a classical function (i.e., a quantum program composed solely of \(\gateX\) and multi-controlled \(\mathrm{NOT}\) gates) safely uncomputes a dirty qubit \(q \in \qubits\) if and only if the following conditions hold for all states \(\rho\):
    \[
    \begin{aligned}
        & \rho|_q = |0\>\<0| &&\implies \sem{C}(\rho)|_q = |0\>\<0|,\mbox{ and} \\
        & \rho|_q = |+\>\<+|  &&\implies \sem{C}(\rho)|_q = |+\>\<+|, \\
    \end{aligned}
    \]
    where, since \(C\) contains no \(\tborrow\) statements, \(\sem{C}\) denotes a single deterministic quantum operation without ambiguity.
\end{theorem}
\begin{proof}
    Since the semantics \(\sem{C}\)  of the circuit implementing a classical function is a unitary transformation \(\sem{C}(\rho)=U\rho U^\dagger\) where \(U\) is a permutation operator, it suffices to show that 
    \[
    U=I_q\otimes V\mbox{ for some }V
    \]
    if and only if for all states \(\rho\), 
    \[ 
    \begin{array}{l}
        \Tr_{\qubits\setminus\{q\}}\pare{U(|0\>_q\<0|\otimes \rho)U^\dagger}= |0\>\<0|,\mbox{ and }\\
        \Tr_{\qubits\setminus\{q\}}\pare{U(|+\>_q\<+|\otimes \rho)U^\dagger}= |+\>\<+|.
    \end{array}
    \]

    (\(\Rightarrow\)) Similar to the proof of (\(\Rightarrow\)) direction in Theorem~\ref{thm:state-restoration}.

    (\(\Leftarrow\)) We take \(\rho=I\) and omit the normalization factor.
    Assume that 
    \[
    U=\begin{pmatrix}
        A_{n\times n} & B_{n\times n} \\
        C_{n\times n} & D_{n\times n}
    \end{pmatrix}_{2n\times 2n},
    \]
    then it holds that
    \[
    \begin{array}{l}
        U(|+\>_q\<+|\otimes I)U^\dagger=\begin{pmatrix}
        (A+B)(A+B)^\dagger & (A+B)(C+D)^\dagger\\
        (C+D)(A+B)^\dagger & (C+D)(C+D)^\dagger
        \end{pmatrix},\\
        U(|0\>_q\<0|\otimes I)U^\dagger=\begin{pmatrix}
        AA^\dagger & AC^\dagger\\
        CA^\dagger & CC^\dagger
        \end{pmatrix},
    \end{array}
    \]
    where coefficient $\frac{1}{2}$ is omitted here. 
    Because
    \[
    Tr_{\qubits\setminus\{q\}}\begin{pmatrix}
        A & B\\
        C & D
    \end{pmatrix}=\begin{pmatrix}
        Tr(A) & Tr(B)\\
        Tr(C) & Tr(D)
    \end{pmatrix}/n,
    \]
    and 
    \[
    \begin{array}{l}
    Tr_{\qubits\setminus\{q\}}(U(|+\>_q\<+|\otimes I)U^\dagger)=|+\>_q\<+|,\\
    Tr_{\qubits\setminus\{q\}}(U(|0\>_q\<0|\otimes I)U^\dagger)=|0\>_q\<0|,
    \end{array}
    \]
    therefore
    $$
    \begin{cases}
        \<A,A\>=n\\
        \<C,C\>=0\\
        \<A,A\>+\<B,B\>+2\<A,B\>=n\\
        \<A,C\>+\<A,D\>+\<B,C\>+\<B,D\>=n\\
    \end{cases},
    $$
    where $\<A,B\>\triangleq Tr(A^\dagger B)=\sum_{i,j=1}^n A_{i,j}B_{i,j}$ is a well-known inner product over the space of linear operators. From the second equation, one knows that 
    $$
    C=\mathbf{0}.
    $$
    Notice that $U$ is a permutation matrix. Thus $$0\leq \<B,B\>,\<A,B\>\leq n,$$ and from the third equation one knows $\<B,B\>=0$, implying that
    $
    B=\mathbf{0}.
    $
    From the last equation, one can verify that $\<A,D\>=n$, and therefore $A=D$.
    In conclusion,
    $$
    U=\begin{pmatrix}
    A & \mathbf{0}\\
    \mathbf{0} & A
    \end{pmatrix}=I_q\otimes A.
    $$
\end{proof}

\subsection{Proof of Theorem~\ref{thm:boolean}}
\begin{theorem}
    A quantum circuit \(C\) implementing a classical function safely uncomputes a dirty qubit \(q\) if and only if both Boolean formulas \eqref{eq:cond-1} and \eqref{eq:cond-2} are unsatisfiable. 
\end{theorem}
\begin{proof}
    From theorem~\ref{thm:safe-uncomp}, we know that safe uncomputation is equivalent to the restoration of two specific states: \(|0\>\<0|\) and \(|+\>\<+|\). 
    Besides, since \(C\) is a circuit implementing some unitary transformation \(U\), it can be verified that the condition 
    \[
    \forall \rho,\ \sem{C}(\rho\otimes |0\>_q\<0|)=\rho'\otimes |0\>_q\<0|\mbox{ for some }\rho'
    \]
    is equivalent to  
    \begin{equation}\label{eq:ac1}
    U|a_1...a_n\>|0\>_q=|a_1'...a_n'\>|0\>_q,
    \end{equation}
    where \(|a_1...,a_n\>\) is a basis pure state of \(\qubits\setminus\{q\}\) in computational basis.
    It is evident that \eqref{eq:ac1} is equivalent to the unsatisfiability of the Boolean formula in \eqref{eq:cond-1}, as the application of \(U\) to the computational basis state can be fully described by the Boolean formulas \(b_q\).

    Similarly, for the condition 
    \[
    \forall \rho,\ \sem{C}(\rho\otimes |+\>_q\<+|)=\rho'\otimes |+\>_q\<+|\mbox{ for some }\rho',
    \]
    it suffices to verify that 
    \begin{equation}\label{eq:ac2}
    U|a_1...a_n\>(|0\>+|1\>)_q=|a_1'...a_n'\>(|0\>+|1\>)_q.
    \end{equation}
    In other words, if considered separately:
    \[
    \begin{cases}
        U|a_1...a_n\>|0\>_q=|a_1'...a_n'\>|0\>_q,\\
        U|a_1...a_n\>|1\>_q=|a_1''...a_n''\>|1\>_q,
    \end{cases}
    \]
    the condition \eqref{eq:ac2} is equivalent to check that \(a_1'...a_n'=a_1''...a_n''\) for all \(a_1...a_n\).
    This is equivalent to the unsatisfiability of the Boolean formula in \eqref{eq:cond-2}.
\end{proof}
\section{Artifact Appendix}
This artifact presents the implementation of a restricted version of \textsf{QBorrow}, a quantum programming language that supports borrowing dirty qubits, along with a verification algorithm for ensuring the safe uncomputation of dirty qubits.
This artifact includes the grammar source code of \textsf{QBorrow}, a benchmark suite, and detailed instructions for installing and executing the verification experiments from the source code.

\subsection{Artifact Check-list}
The implementation primarily includes the following files
\begin{itemize}
    \item \texttt{QBorrow/CMakeLists.txt}: the CMake configuration file for verifying the presence of required files and libraries, including the ANTLR4 jar file, the ANTLR4 C++ runtime, and the C++ APIs for CVC5 and Bitwuzla.
    \item \texttt{QBorrow/grammar/QBorrow.g4}: the grammar of \textsf{QBorrow} in ANTLR g4 format.
    \item \texttt{QBorrow/include/frontend/antlr-gen/*}, 
    
    \texttt{QBorrow/src/frontend/antlr-gen/*}: the ANTLR-generated C++ parser and lexer for \textsf{QBorrow}. Provided the grammar file, the CMake script can automatically generate these files if they are not present.
    \item \texttt{QBorrow/examples/adder.qbr}, 
    
    \texttt{QBorrow/examples/mcx.qbr}: two benchmark programs implemented in \textsf{QBorrow} for verification experiments.

    \item Other files in \texttt{QBorrow/include} and \texttt{QBorrow/src} implement the AST construction, Boolean formula generation and SMT solving for the verification algorithm.
\end{itemize}

\subsection{Instructions for Installation and Execution}
To build from source, one needs to install the following dependencies:
\begin{itemize}
    \item the jar file for ANTLR4;
    \item the C++ runtime for ANTLR4;
    \item the C++ APIs for CVC5 and Bitwuzla. 
\end{itemize}
One may manually generate the ANTLR4 C++ parser and lexer from the grammar file using the command
\begin{lstlisting}[style=qborrow]
java -jar /path/to/antlr-4.13.x-complete.jar \ 
-Dlanguage=Cpp -visitor -o \ 
/path/to/output/dir /path/to/QBorrow.g4
\end{lstlisting}
while the generation can also be automatically handled by the CMake script during the build process.

To build the project from source, one may need either hard-code the paths of the above dependencies in 

\noindent
\texttt{QBorrow/CMakeLists.txt}, or explicitly provide them when executing the CMake command, e.g.,
\begin{lstlisting}[style=qborrow]
mkdir build && cd build

cmake -DANTLR4_JAR_PATH\ 
    =/path/to/antlr-4.13.x-complete.jar \ 
-DANTLR4_RUNTIME_PREFIX_PATH\ 
    =/path/to/antlr4/runtime/cpp \ 
-DCVC5_PREFIX_PATH=/path/to/cvc5 \ 
-DBITWUZLA_PREFIX_PATH=/path/to/bitwuzla \
..
\end{lstlisting}
Once the configuration is done, one can build the project using the command
\begin{lstlisting}[style=qborrow]
make
\end{lstlisting}
and there should be an executable file named \texttt{qborrow} in the build directory.
For MacOS users, there is already a pre-compiled binary file in the \texttt{QBorrow/build} directory, which can be directly executed without building from source.
To run the verification experiments on the benchmark programs, one can use the command
\begin{lstlisting}[style=qborrow]
make adder 
\end{lstlisting}
or manually execute the command
\begin{lstlisting}[style=qborrow]
./qborrow ../examples/adder.qbr
\end{lstlisting}

\subsection{Grammar of \textsf{QBorrow} in ANTLR g4}
The grammar file of \textsf{QBorrow} for ANTLR4 is provided below.  
It specifies the syntax of the quantum programming language, encompassing constructs for borrowing dirty qubits, allocating and releasing qubits, as well as fundamental quantum operations such as X, CNOT, and CCNOT gates.  
Additionally, to enhance programmer convenience, it supports basic for loops and primitive arrays of qubits.

\begin{lstlisting}[style=qborrow]
grammar QBorrow;

program : statement+ EOF;

statement 
    : 'let' ID '=' expr ';'
    | 'borrow' reg ';'
    | 'borrow@' reg ';'
    | 'alloc' reg ';'
    | 'release' ID ';'
    | 'X' '[' reg ']' ';'
    | 'CNOT' '[' reg ',' reg ']' ';'
    | 'CCNOT' '[' reg ',' reg ',' reg ']' ';'
    | 'for' ID '=' expr 'to' expr '{' statement* '}'
    ;

reg
    : ID '[' expr ']'
    | ID 
    ;

// arithmetic expressions
expr 
    : expr ADD term 
    | expr SUB term
    | SUB term
    | ADD term
    | term 
    ;


term 
    : term MUL factor 
    | factor
    ;

factor
    : NUMBER
    | ID
    | '(' expr ')'
    ;

ADD : '+';
SUB : '-';
MUL : '*';

ID : [a-zA-Z_] [a-zA-Z0-9_]* ;
NUMBER : [0-9]+ ;
WS : [ \t\r\n]+ -> skip ;

LINE_COMMENT  : '//' ~[\r\n]*     -> skip ;
BLOCK_COMMENT : '/*' .*? '*/'     -> skip ;
\end{lstlisting}

\subsection{Benchmark Programs}
The two benchmark programs are the implementation of the adder circuit in~\cite{haner2016factoring} and the multi-controlled NOT gate (MCX) circuit in~\cite{Gidney-blog}.

Specifically, the adder circuit is illustrated in Figure~\ref{fig:circuit-adder}, and its implementation in \textsf{QBorrow} can be found in the file \texttt{QBorrow/examples/adder.qbr} as shown below:
\begin{lstlisting}[style=qborrow]
 // adder.qbr
 let n = 50;        // number of qubits
 borrow@ q[n];      // skip verification
 borrow a[n - 1];   // dirty qubits
 CNOT[a[n - 1], q[n]];
 for i = (n - 1) to 2 {
     CNOT[q[i], a[i]];
     X[q[i]];
     CCNOT[a[i - 1], q[i], a[i]];
 }
 CNOT[q[1], a[1]];
 for i = 2 to (n - 1) {
     CCNOT[a[i - 1], q[i], a[i]];
 }
 CNOT[a[n - 1], q[n]];
 X[q[n]];

 // reverse the circuit to uncompute
 for i = (n - 1) to 2 {
     CCNOT[a[i - 1], q[i], a[i]];
 }
 CNOT[q[1], a[1]];
 for i = 2 to (n - 1) {
     CCNOT[a[i - 1], q[i], a[i]];
     X[q[i]];
     CNOT[q[i], a[i]];
 }
    \end{lstlisting}

\begin{figure*}[t]
    \centering
    \def\udots{\cdot^{\cdot^{\cdot}}}
    \begin{quantikz}[row sep = 0.7em, column sep = 0.3em]
        \lstick{\(q_1\)}     & \qw      & \qw      & \qw      & \qw      &\qw& \cdots && \qw      & \qw      & \qw      & \ctrl{1} & \qw      &\qw& \cdots && \qw      & \qw      & \qw      &\qw&\cdots&& \qw      & \ctrl{1} & \qw      & \qw      & \qw      &\qw&\cdots&&\qw      & \qw      & \qw      &\qw& \rstick{\(q_1\)}\\
        \lstick{\(a_1\)}     & \qw      & \qw      & \qw      & \qw      &\qw& \cdots && \qw      & \qw      & \ctrl{2} & \targ{}  & \ctrl{2} &\qw& \cdots && \qw      & \qw      & \qw      &\qw&\cdots&& \ctrl{2} & \targ{}  & \ctrl{2} & \qw      & \qw      &\qw&\cdots&&\qw      & \qw      & \qw      &\qw& \rstick{\(a_1\)}\\
        \lstick{\(q_2\)}     & \qw      & \qw      & \qw      & \qw      &\qw& \cdots && \ctrl{1} & \targ{}  & \ctrl{1} & \qw      & \ctrl{1} &\qw& \cdots && \qw      & \qw      & \qw      &\qw&\cdots&& \ctrl{1} & \qw      & \ctrl{1} & \targ{}  & \ctrl{1} &\qw&\cdots&&\qw      & \qw      & \qw      &\qw& \rstick{\(q_2\)}\\
        \lstick{\(a_2\)}     & \qw      & \qw      & \qw      & \qw      &\qw& \cdots && \targ{}  & \qw      & \targ{}  & \qw      & \targ{}  &\qw& \cdots && \qw      & \qw      & \qw      &\qw&\cdots&& \targ{}  & \qw      & \targ{}  & \qw      & \targ{}  &\qw&\cdots&&\qw      & \qw      & \qw      &\qw& \rstick{\(a_2\)}\\
        \lstick{\(\vdots\)}  &          &          &          &          &   & \udots &&          &          &          &          &          &   & \ddots &&          &          &          &   &\udots&&          &          &          &          &          &   &\ddots&&         &          &          &   & \rstick{\(\vdots\)}\\
        \lstick{\(a_{n-2}\)} & \qw      & \qw      & \qw      & \ctrl{2} &\qw& \cdots && \qw      & \qw      & \qw      & \qw      & \qw      &\qw& \cdots && \ctrl{2} & \qw      & \ctrl{2} &\qw&\cdots&& \qw      & \qw      & \qw      & \qw      & \qw      &\qw&\cdots&&\ctrl{2} & \qw      & \qw      &\qw& \rstick{\(a_{n-2}\)}\\
        \lstick{\(q_{n-1}\)} & \qw      & \ctrl{1} & \targ{}  & \ctrl{1} &\qw& \cdots && \qw      & \qw      & \qw      & \qw      & \qw      &\qw& \cdots && \ctrl{1} & \qw      & \ctrl{1} &\qw&\cdots&& \qw      & \qw      & \qw      & \qw      & \qw      &\qw&\cdots&&\ctrl{1} & \targ{}  & \ctrl{1} &\qw& \rstick{\(q_{n-1}\)}\\
        \lstick{\(a_{n-1}\)} & \ctrl{1} & \targ{}  & \qw      & \targ{}  &\qw& \cdots && \qw      & \qw      & \qw      & \qw      & \qw      &\qw& \cdots && \targ{}  & \ctrl{1} & \targ{}  &\qw&\cdots&& \qw      & \qw      & \qw      & \qw      & \qw      &\qw&\cdots&&\targ{}  & \qw      & \targ{}  &\qw& \rstick{\(a_{n-1}\)}\\
        \lstick{\(q_n\)}     & \targ{}  & \qw      & \qw      & \qw      &\qw& \cdots && \qw      & \qw      & \qw      & \qw      & \qw      &\qw& \cdots && \qw      & \targ{}  & \qw      &\qw&\cdots&& \qw      & \qw      & \qw      & \qw      & \qw      &\qw&\cdots&&\qw      & \qw      & \qw      &\qw& \rstick{\(\tilde{q_{n-1}}\)}
    \end{quantikz}
    \caption{Circuit for the adder program adapted from~\cite{haner2016factoring}.}
    \label{fig:circuit-adder}

    ~\\
    \begin{tabular}{c|ccccccc}
        \mbox{Duration (s)} & 50 qubits & 75 qubits & 100 qubits & 125 qubits & 150 qubits & 175 qubits & 200 qubits \\ \hline
        \mbox{CVC5} & 4 & 24 & 71 & 171 & 365 & 751 & 1069 \\
        \mbox{Bitwuzla} & 3 & 12 & 29 & 98 & 158 & 248 & 313
    \end{tabular}
    \caption{Verification time for the adder program with different qubit numbers using CVC5 and Bitwuzla.}
    \label{fig:exp-adder}
    ~\\
    \begin{tabular}{c|ccccccc}
        \mbox{Duration (s)} & 499 qubits & 999 qubits & 1499 qubits & 1999 qubits & 2499 qubits & 2999 qubits & 3499 qubits \\ \hline
        \mbox{CVC5} & 0 & 1 & 4 & 7 & 11 & 17 & 27 \\
        \mbox{Bitwuzla} & 3 & 16 & 35 & 61 & 115 & 163 & 239
    \end{tabular}
    \caption{Verification time for the MCX program with different qubit numbers using CVC5 and Bitwuzla.}
    \label{fig:exp-mcx}
\end{figure*}

The second benchmark program implements the multi-controlled NOT gate (MCX) circuit using dirty qubits, as illustrated in the following Figure~\ref{fig:circuit-mcx}.
\begin{figure}[H]
    \centering
    \includegraphics[width=\linewidth]{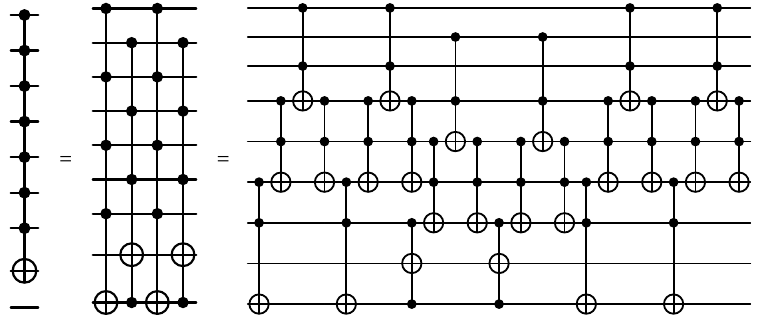}
    \caption{Circuit for the MCX program adapted from~\cite{Gidney-blog}.}
    \label{fig:circuit-mcx}
\end{figure}

Its implementation in \textsf{QBorrow} can be found in the file \texttt{QBorrow/examples/mcx.qbr} as shown below:
\begin{lstlisting}[style=qborrow]
// mcx.qbrs
let m = 1750;
let n = m + (m - 1); // n-controlled NOT gate

borrow@ q[n];
borrow@ t;

borrow anc;

// first part
CCNOT[q[n - 1], q[n], anc];

for i = (m - 2) to 2 {
    CCNOT[q[2 * i - 1], q[2 * i + 1], q[2 * i + 2]];
}

CCNOT[q[1], q[3], q[4]];

for i = 2 to (m - 2) {
    CCNOT[q[2 * i - 1], q[2 * i + 1], q[2 * i + 2]];
}

CCNOT[q[n - 1], q[n], anc];

for i = (m - 2) to 2 {
    CCNOT[q[2 * i - 1], q[2 * i + 1], q[2 * i + 2]];
}

CCNOT[q[1], q[3], q[4]];

for i = 2 to (m - 2) {
    CCNOT[q[2 * i - 1], q[2 * i + 1], q[2 * i + 2]];
}

// second part

CCNOT[q[n], anc, t];

for i = (m - 1) to 3 {
    CCNOT[q[2 * i - 1], q[2 * i], q[2 * i + 1]];
}

CCNOT[q[2], q[4], q[5]];

for i = 3 to (m - 1) {
    CCNOT[q[2 * i - 1], q[2 * i], q[2 * i + 1]];
}

CCNOT[q[n], anc, t];

for i = (m - 1) to 3 {
    CCNOT[q[2 * i - 1], q[2 * i], q[2 * i + 1]];
}

CCNOT[q[2], q[4], q[5]];

for i = 3 to (m - 1) {
    CCNOT[q[2 * i - 1], q[2 * i], q[2 * i + 1]];
}


// third part

CCNOT[q[n - 1], q[n], anc];

for i = (m - 2) to 2 {
    CCNOT[q[2 * i - 1], q[2 * i + 1], q[2 * i + 2]];
}

CCNOT[q[1], q[3], q[4]];

for i = 2 to (m - 2) {
    CCNOT[q[2 * i - 1], q[2 * i + 1], q[2 * i + 2]];
}

CCNOT[q[n - 1], q[n], anc];

for i = (m - 2) to 2 {
    CCNOT[q[2 * i - 1], q[2 * i + 1], q[2 * i + 2]];
}

CCNOT[q[1], q[3], q[4]];

for i = 2 to (m - 2) {
    CCNOT[q[2 * i - 1], q[2 * i + 1], q[2 * i + 2]];
}

// fourth part 

CCNOT[q[n], anc, t];

for i = (m - 1) to 3 {
    CCNOT[q[2 * i - 1], q[2 * i], q[2 * i + 1]];
}

CCNOT[q[2], q[4], q[5]];

for i = 3 to (m - 1) {
    CCNOT[q[2 * i - 1], q[2 * i], q[2 * i + 1]];
}

CCNOT[q[n], anc, t];

release anc;

for i = (m - 1) to 3 {
    CCNOT[q[2 * i - 1], q[2 * i], q[2 * i + 1]];
}

CCNOT[q[2], q[4], q[5]];

for i = 3 to (m - 1) {
    CCNOT[q[2 * i - 1], q[2 * i], q[2 * i + 1]];
}
\end{lstlisting}

\subsection{Experimental Results}
We conducted verification experiments on a MacBook Air with an M3 chip and 24GB of RAM.
The results are summarized in the Figure~\ref{fig:exp-adder} and Figure~\ref{fig:exp-mcx}.

\newpage

\end{document}